\documentclass[12pt,a4paper,amsmath,amssymb,reqno]{amsart}
\usepackage[latin1]{inputenc}
\usepackage{amsthm}
\usepackage{latexsym}
\usepackage{amsfonts}
\usepackage{bbm,bm,dsfont}
\usepackage{eufrak}
\usepackage{color} 
\usepackage{graphicx}

\newtheorem{proposition}{Proposition}
\newtheorem{proposition?}{Proposition?}
\newtheorem{theorem}{Theorem}
\newtheorem{lemma}{Lemma}
\newtheorem{corollary}{Corollary}

\theoremstyle{definition}

\newtheorem{example}{Example}
\newtheorem{definition}{Definition}

\newcommand{\mc}[1]{\mathcal{#1}}
\newcommand{\msf}[1]{\mathsf{#1}}
\newcommand{\mr}[1]{\mathrm{#1}}
\newcommand{\mfr}[1]{\mathfrak{#1}}

\newcommand{\N}{\mathbb N}

\newcommand{\C}{\mathbb C}

\newcommand{\half}{\tfrac{1}{2}} %half
\newcommand{\mo}[1]{\left| #1 \right|} %modulus

\newcommand{\hil}{\mathcal{H}} %Hilbert space

\newcommand{\hi}{\mathcal{H}} %Hilbert space H
\newcommand{\hik}{\mathcal{K}} %Hilbert space K
\newcommand{\trh}{\mathcal{T(H)}} %trace class operators on H
\newcommand{\ip}[2]{\left\langle\,#1\,|\,#2\,\right\rangle} %inner product
\newcommand{\kb}[2]{|#1\rangle\langle#2|} %ketbra
\newcommand{\tr}[1]{\mathrm{tr}\left[#1\right]} %trace
\newcommand{\id}{\mathbbm{1}} %identity operator
\newcommand{\Ao}{\mathsf{A}}%generic observable
\newcommand{\Bo}{\mathsf{B}}%generic observable
\newcommand{\Po}{\mathsf{P}}%sharp observable
\newcommand{\Ii}{\mathcal{I}} %generic instrument
\newcommand{\Li}{\mathcal{L}} %Luders instrument

\def\<{\langle}
\def\>{\rangle}

\newcommand{\fii}{\varphi}

  \makeatletter
  \def\mathcomposite{%
     \@ifstar
        {\def\@mathcomposite@option{%
            \baselineskip\z@skip\lineskiplimit-\maxdimen}%
         \@mathcomposite}%
        {\let\@mathcomposite@option\offinterlineskip
         \@mathcomposite}}
  \def\@mathcomposite{%
     \@ifnextchar[\@@mathcomposite{\@@mathcomposite[0]}}
  \def\@@mathcomposite[#1]#2#3#4{%
     #2{\mathchoice
        {\@mathcomposite@{#1}{#3}{#4}\displaystyle{1}}%
        {\@mathcomposite@{#1}{#3}{#4}\textstyle{1}}%
        {\@mathcomposite@{#1}{#3}{#4}%
         \scriptstyle\defaultscriptratio}%
        {\@mathcomposite@{#1}{#3}{#4}%
         \scriptscriptstyle\defaultscriptscriptratio}}}
  \def\@mathcomposite@#1#2#3#4#5{%
     \vcenter{\m@th\@mathcomposite@option
        \dimen@\f@size\p@\dimen@#1\dimen@\dimen@#5\dimen@
        \divide\dimen@ 18
        \edef\@mathcomposite@skipamount{\the\dimen@}%
        \ialign{\hfil$#4##$\hfil\cr
           #2\crcr
           \noalign{\vskip\@mathcomposite@skipamount}%
           #3\crcr}}}
  \makeatother

\newcommand{\psleq}{\mathcomposite{\mathrel}{\prec}{\sim}}

\begin{document}

\title[]{The unavoidable information flow to environment in quantum measurements}

\author[]{Erkka Haapasalo$^\diamondsuit$}
\address{$\diamondsuit$Erkka Haapasalo, Department of Nuclear Engineering, Kyoto University, 6158540 Kyoto, Japan}
\email{haapasalo.erkka.45e@st.kyoto-u.ac.jp}

\author[]{Teiko Heinosaari$^\clubsuit$}
\address{$\clubsuit$Teiko Heinosaari, Turku Centre for Quantum Physics, Department of Physics and Astronomy, University of Turku, Finland}
\email{teiko.heinosaari@utu.fi}

\author[]{Takayuki Miyadera$^\spadesuit$}
\address{$\spadesuit$Takayuki Miyadera, Department of Nuclear Engineering, Kyoto University, 6158540 Kyoto, Japan}
\email{miyadera@nucleng.kyoto-u.ac.jp}

\begin{abstract}
One of the basic lessons of quantum theory is that one cannot obtain information on an unknown quantum state without disturbing it. Hence, by performing a certain measurement, we limit the other possible measurements that can be effectively implemented on the original input state. It has been recently shown that one can implement sequentially any device, either channel or observable, which is compatible with the first measurement \cite{HeMi15}. In this work we prove that this can be done, apart from some special cases, only when the succeeding device is implemented on a larger system than just the input system. This means that some part of the still available quantum information has been flown to the environment and cannot be gathered by accessing the input system only.
We characterize the size of the post-measurement system by determining the class of measurements for the observable in question that allow the subsequent realization of any measurement process compatible with the said observable. We also study the class of measurements that allow the subsequent realization of any observable jointly measurable with the first one and show that these two classes coincide when the first observable is extreme.
\end{abstract}

\maketitle

%%%%%%%%%%%%%%%%%%%%
\section{Introduction}
%%%%%%%%%%%%%%%%%%%%
 
In a measurement of a quantum observable $\msf A$, the obtained result can be divided into two forms: into the classical information $p_\varrho^\Ao$, where $p_\varrho^\Ao(j)$ is the probability of detecting the value $j$ of the observable in the input state $\varrho$, and into the quantum information remaining in the post-measurement system after the measurement. In this work we are interested in this remaining quantum information and we ask the question: \emph{how can we measure $\msf A$ to enable as many different subsequent measurements or other quantum information processing protocols as possible?}

A partial answer has been given earlier by two of the authors of the present paper. As shown in \cite{HeMi15}, there is a way to measure $\msf A$ in such a way that it allows the subsequent realization of any device $\mathcal{D}$ compatible with $\msf A$ in a sequential set-up. This means that there is a device $\mathcal{D}'$ such that, carrying out $\mathcal{D}'$ after the specific least disturbing measurement of $\msf A$, one has an implementation of $\mathcal{D}$. Here $\mathcal{D}$ can be any other observable or a quantum channel compatible with $\msf A$.  All this means that the unconditional state transformation associated with this least disturbing measurement of $\Ao$, namely the least disturbing channel $\Lambda_{\msf A}$ associated with $\msf A$, preserves all the quantum information in the post-measurement system required to perform any measurement processes compatible with $\msf A$ after the least disturbing measurement.

All the channels that are concatenation equivalent with the least disturbing channel have this same information preservation property. In this paper we give a complete characterization of this equivalence class: it  consists exactly all the least disturbing channels $\Lambda_{\msf B}$ associated with observables $\msf B$ that are post-processing equivalent with $\msf A$. An important tool in our proofs is the recent definition of minimally sufficient observables \cite{Kuramochi15b}. This notion enables us to further pinpoint a particular representative from the class of least disturbing channel class; the one with the lowest output dimension, thus being the least wasteful least disturbing channel associated with a given observable. In this way, we get a characterization of how large environment one has to minimally take into account if one does not want to loose any information that can be possibly used after the measurement of $\Ao$. Since in most cases the auxiliary device $\mathcal{D}'$ must operate in a larger system than $\mathcal{D}$, one can conclude that the information flow to environment is unavoidable.

We also study a strictly smaller subclass of sequential schemes described above: sequential measurements. Now we only ask how to measure an observable $\msf A$ so that any joint measurement of $\msf A$ with some other observable $\msf B$ can be carried out by first measuring $\msf A$ in this special way followed by measuring some possibly distorted version $\msf B'$ of $\msf B$. The unconditional state transformation of this special measurement of $\msf A$ should leave enough information in the post-measurement system to enable the subsequent realizations of measurements of all observables jointly measurable with $\msf A$. Note that a priori we may not need the full power of the least disturbing channel $\Lambda_{\msf A}$ for this task since we are now looking only at subsequent realizations of observables, not those of all the more general quantum information processing tasks described by channels. However, we prove that, at least when $\msf A$ is an extreme observable, these tasks coincide; we need the whole information preservation power of $\Lambda_\Ao$ for the sequential measurement scenario.

\begin{center}
\begin{figure}
\includegraphics[scale=0.17]{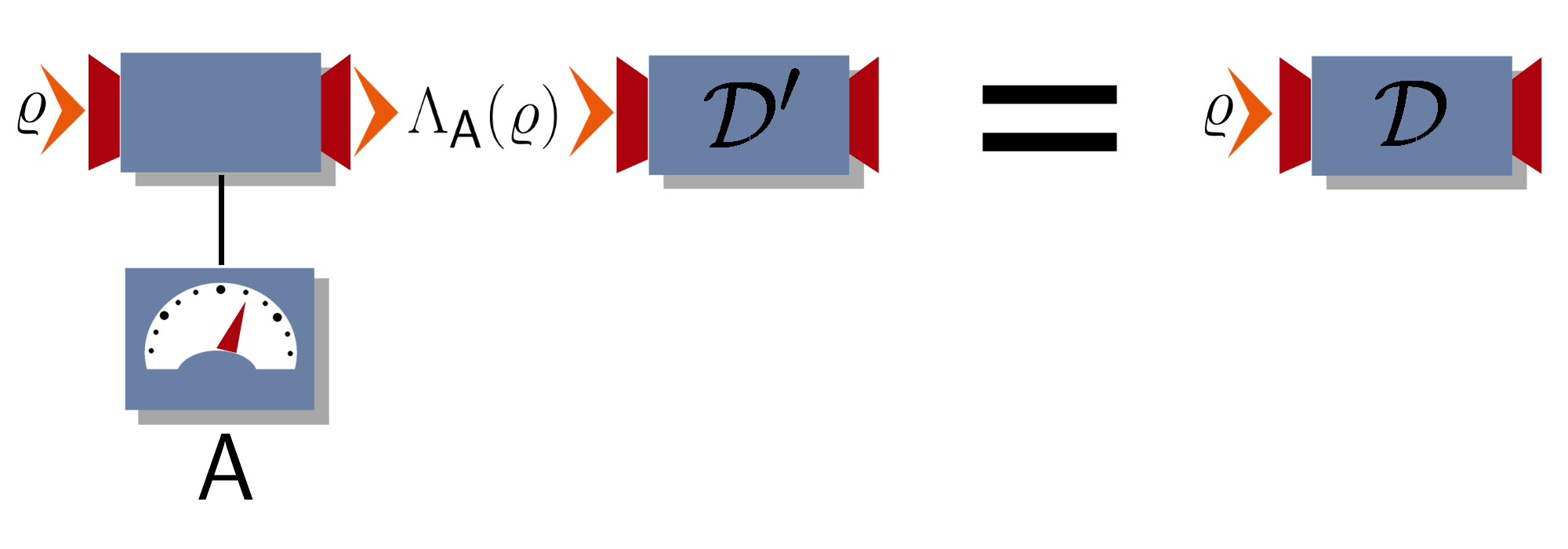}
\caption{For any observable $\msf A$, there is a way to measure $\msf A$ such that any device $\mc D$ compatible with $\msf A$ can be realized by some device $\mc D'$ after the said measurement of $\msf A$. When the unconditioned state transformation induced by such a universal measurement is denoted by $\Lambda_\Ao$, this means that $\mc D=\mc D'\circ\Lambda_{\msf A}$.\label{fig:output1}}
\end{figure}
\end{center}

\begin{center}
\begin{figure}
\includegraphics[scale=0.17]{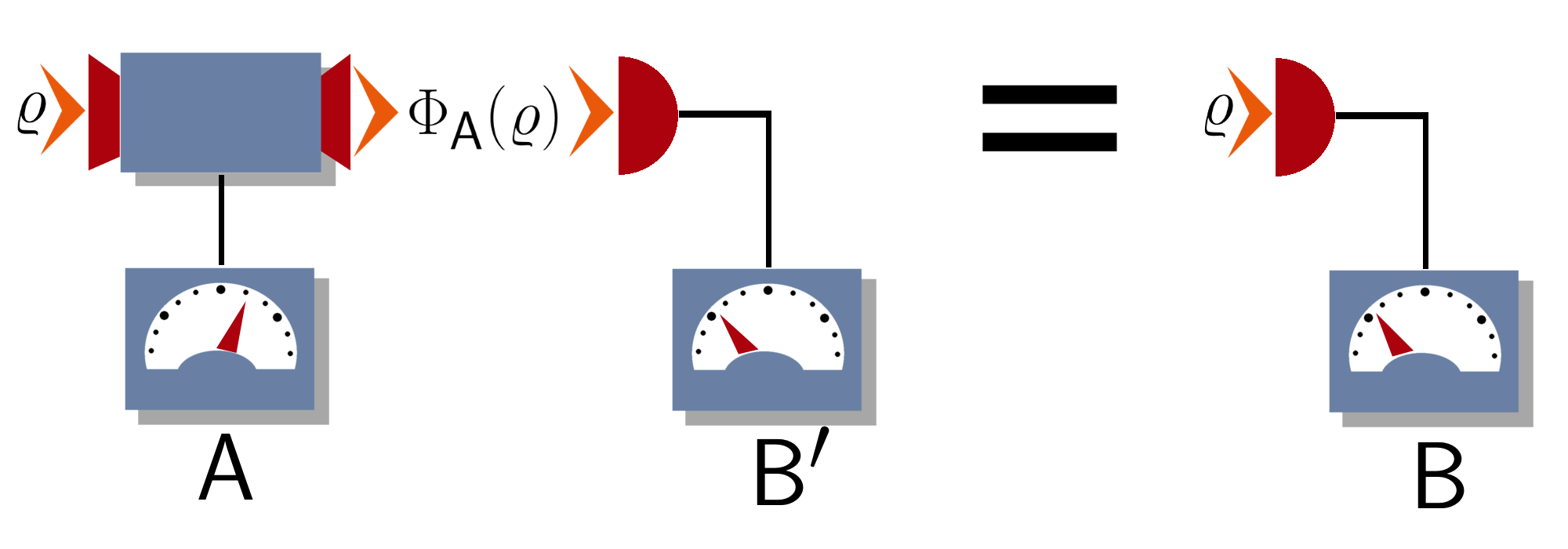}
\caption{We may measure any observable $\msf B$ compatible (jointly measurable) with a fixed observable $\msf A$ always in a sequential setting: We first measure $\msf A$ in a particular way such that the unconditioned state transformation induced by this measurement is $\Phi_\Ao$ and any observable $\msf B$ compatible with $\msf A$ can be realized by measuring some observable $\msf B'$ after the particular measurement of $\msf A$. It is clear that a measurement of $\msf A$ inducing the channel $\Lambda_{\msf A}$ of Fig. \ref{fig:output1} can do this, but is $\Lambda_{\msf A}$ actually necessary? \label{fig:output2}}
\end{figure}
\end{center}

%%%%%%%%%%%%%%%%%%%%
\section{Sequential quantum measurements}
%%%%%%%%%%%%%%%%%%%%

Before formulating and proving the results outlined above, we recall some basics of quantum measurements by defining the notions of observables, channels, and instruments. An important notion throughout this work is compatibility of an observable with channels and with other observables. We also discuss the link between sequential measurements and compatibility and give two (a priori) different definitions of universality of a channel with respect to an observable. These universality properties are at the core of our subsequent discussion.

%%%%%%%%%%%%%%%%%%%%
\subsection{Quantum observables, channels and instruments}
%%%%%%%%%%%%%%%%%%%%

As usual, we identify the physical observables of a quantum system with positive-operator-valued measures (POVMs). We shall concentrate on discrete observables. In this view, an observable on a system described by a Hilbert space $\hil$ and with outcomes in $\Omega_{\msf A}\subset\N$ is thus a positive-operator-valued map $\msf A:\Omega_{\msf A}\to\mc L(\hil)$ such that $\sum_{j\in\Omega_{\msf A}}\msf A(j)=\id_\hil$. A POVM $\msf P:\Omega_{\msf P}\to\mc L(\hil)$, $\Omega_{\msf P}\subset\N$, is a projection-valued measure (PVM) if $\msf P(j)$ is a projection for each $j\in\Omega_{\msf P}$. The physical observables corresponding to PVMs are called as {\it sharp observables}. We have introduced the notations $\mc L(\hil)$ for the set of bounded operators on $\hil$ and $\id_\hil$ for the identity operator on $\hil$. Moreover, let us denote the set of states on $\hil$, i.e.,\ positive operators on $\hil$ of trace 1, by $\mc S(\hil)$. For any state $\varrho\in\mc S(\hil)$, the observable $\msf A$ determines a probability distribution $p_\varrho^{\msf A}:\Omega_{\msf A}\to[0,1]$ through
\begin{equation}
p_\varrho^{\msf A}(j):=\tr{\varrho\msf A(j)} \, ,\qquad j\in\Omega_{\msf A} \, .
\end{equation}
The number $p_\varrho^{\msf A}(j)$ is the probability of detecting the outcome $j$ in any measurement of $\msf A$ when the system was initially in the state $\varrho$.

Let $\hil$ and $\mc K$ be Hilbert spaces and denote by $\trh$ (resp. $\mc T(\mc K)$) the trace class on $\hil$ (resp. $\mc K$). A linear map $\Lambda:\mc S(\hil)\to\mc S(\mc K)$ is called an {\it operation} when its dual $\Lambda^*:\mc L(\mc K)\to\mc L(\hil)$ defined through 
\begin{equation}
\tr{\varrho\Lambda^*(B)}=\tr{\Lambda(\varrho)B} \, , \quad \varrho\in\mc S(\hil),  B\in\mc L(\mc K) \, , 
\end{equation}
is completely positive and $\Lambda^*(\id_{\mc K})\leq\id_\hil$. The latter condition is equivalent with $\tr{\Lambda(\varrho)}\leq1$ for all $\varrho\in\mc S(\hil)$. A {\it channel} is an operation $\Lambda$ with $\Lambda^*(\id_{\mc K})=\id_\hil$ or, equivalently, $\Lambda$ maps states into states. Hence, channels describe state transformations that can be parts of the time evolution of the system or, in the case we are now particularly interested in, induced by a measurement of a quantum system in which the system before the measurement is described by the Hilbert space $\hil$ and the post-measurement system is associated with $\mc K$. Naturally, a channel $\Lambda:\mc T(\hil)\to\mc T(\mc K)$ can be identified with its restriction $\Lambda|_{\mc S(\hil)}:\mc S(\hil)\to\mc S(\mc K)$ which is an affine map.

A detailed mathematical description of a quantum measurement contains both an observable and a channel as its parts. A suitable concept is that of an instrument \cite{QTOS76,DaLe70}.

\begin{definition}
An observable $\Ao:\Omega_\Ao\to\mc L(\hil)$ and a channel $\Lambda:\mc S(\hil)\to\mc S(\mc K)$, where $\mc K$ is any Hilbert space, are \emph{compatible} if there exists an map (an \emph{instrument}) $\mc I:\Omega_{\Ao}\times\mc T(\hil)\to\mc T(\mc K)$ such that, for any $j\in\Omega_{\Ao}$, the map $\mc I(j,\cdot)$ is an operation, $\tr{\mc I(j,\varrho)}=p^\Ao_\varrho(j)$ for all $\varrho\in\mc S(\hil)$, and $\sum_{j\in\Omega_{\Ao}}\mc I(j,\cdot)=\Lambda$. If $\Ao$ and channel $\Lambda$ are not compatible, then they are called \emph{incompatible}. We denote the class of all channels $\Lambda:\mc S(\hil)\to\mc S(\mc K)$ (with varying Hilbert space $\mc K$) compatible with an observable $\Ao:\Omega_\Ao\to\mc L(\hil)$ by $\mfr C_\Ao$.
\end{definition}

The most paradigmatic example of incompatibility is the no-information-without-disturbance theorem. It says that a unitary channel is incompatible with any nontrivial observable, i.e., any observable which is not of the form $\Ao(j)=p_j \id$ for some probability distribution $(p_j)$; see e.g. \cite[Sec. 5.2.2]{MLQT12}. To give an example of compatibility, we recall that the L\"uders instrument $\Li^\Ao$ of $\Ao$ is defined as $\Li^\Ao(j,B)=\sqrt{\Ao(j)} B \sqrt{\Ao(j)}$. Further demonstrations of compatibility between channels and observables can be found in \cite{HeReRyZi17}.

%%%%%%%%%%%%%%%%%%%%
\subsection{Sequential implementation of compatible quantum devices}
%%%%%%%%%%%%%%%%%%%%

According to quantum theory, measuring two observables simultaneously is not always possible. However, when this is possible, we say that the observables are \emph{compatible} or \emph{jointly measurable}. 
Next we give a formal definition of this notion \cite{LaPu97}.

\begin{definition}
We say that observables $\msf A:\Omega_{\msf A}\to\mc L(\hil)$ and $\msf B:\Omega_{\msf B}\to\mc L(\hil)$ are {\it compatible} if there is an observable $\msf G:\Omega_{\msf A} \times \Omega_{\msf B}\to\mc L(\hil)$ such that
\begin{equation}
\begin{array}{rcll}
\msf A(j)&=&\sum_{k\in\Omega_{\msf B}}\msf G(j,k) \, , &j\in\Omega_{\msf A}\\
\msf B(k)&=&\sum_{j\in\Omega_{\msf A}}\msf G(j,k) \, , &k\in\Omega_{\msf B} \, .
\end{array}
\end{equation}
Such an observable $\msf{G}$ is called as a {\it joint observable for $\msf A$ and $\msf B$} and $\msf A$ and $\msf B$ are called as the {\it margins of $\msf G$}. We denote the class of all observables $\msf B:\Omega_{\msf B}\to\mc L(\hil)$ (with a varying value space) compatible with an observable $\msf A:\Omega_{\msf A}\to\mc L(\hil)$ by $\mfr O_\Ao$.
\end{definition}

One possibility to implement a joint observable is a sequential measurement. Suppose that we measure an observable $\msf A:\Omega_{\msf A}\to\mc L(\hil)$, i.e.,\ carry out an {\it $\msf A$-instrument} $\mc I:\Omega_{\msf A}\times\mc T(\hil)\to\mc T(\mc K)$ with some output Hilbert space $\mc K$ such that,  $\tr{\mc I(\cdot,\varrho)}=p^{\msf A}_\varrho$ for all $\varrho\in\mc S(\hil)$. The conditional state after the measurement conditioned by detecting the value $i\in\Omega_{\msf A}$ when the pre-measurement state was $\varrho$ is $\mc I(i,\varrho)$. When we now measure a second observable $\msf C:\Omega_{\msf C}\to\mc L(\mc K)$, we obtain a measurement of an observable $\msf G:\Omega_{\msf A}\times\Omega_{\msf C}\to\mc L(\hil)$ with measurement outcome statistics $p_\varrho^{\msf G}(i,j)=\tr{\mc I(i,\varrho)\msf C(j)}$, i.e.,
\begin{equation}
\msf G(i,j)=[\mc I(i,\cdot)^*]\big(\msf C(j)\big),\qquad i\in\Omega_{\msf A},\quad j\in\Omega_{\msf C},
\end{equation}
with the margins $\msf A=\sum_j\msf G(\cdot,j)$ and $\msf B:=\sum_i\msf G(i,\cdot)=\Lambda^*\circ\msf C$. Hence, a sequential measurement of $\msf A$ and $\msf C$ (in this order) gives rise to a joint measurement of $\msf A$ and a disturbed version $\msf B=\Lambda\circ\msf C$ of $\msf C$ affected by the total channel $\Lambda$ induced by the measurement of $\msf A$.

A question arises, whether any joint measurements of two compatible observables can be realized as a sequential measurement as described above. The answer is quite easily seen to be `yes' \cite{HeWo10}, even in the case of all physically meaningful continuous observables (those with a countably generated value space $\sigma$-algebra) \cite{HaPe17}. However, one can also ask whether, for an observable $\msf A$, there is a measurement of $\msf A$, i.e.,\ an $\msf A$-instrument $\mc I$ such that any observable $\msf B$ compatible with $\msf A$ can be measured jointly with $\msf A$ in a sequential setting by measuring $\msf A$ first with the measurement setting $\mc I$ and then measuring some observable $\msf C$ depending on $\msf B$. As this property depends only on the channel defined by $\Ii$ and not other details of $\Ii$, we formalize it as a property of channels compatible with $\msf A$.

\begin{definition}\label{def:obsuniversal}
Let $\msf A:\Omega_{\msf A}\to\mc L(\hil)$ be an observable. A channel $\Phi:\mc S(\hil)\to\mc S(\hik_\Phi)$ is {\it $\mfr O_\Ao$-universal for $\msf A$} if it is compatible with $\msf A$ and, whenever $\msf B:\Omega_{\msf B}\to\mc L(\hil)$ is an observable compatible with $\msf A$, there is an observable ${\msf B}':\Omega_{\msf B}\to\mc L(\hik_\Phi)$ such that $\msf B=\Phi^*\circ {\msf B}'$.
\end{definition}

Instead of performing a joint measurement of $\Ao$ and some other observable $\Bo$, we may want to apply some channel $\Lambda$ after $\Ao$. Again, we can use some other channel $\Lambda'$ than $\Lambda$, but we want the effective channel to be $\Lambda$.

\begin{definition}\label{def:chuniversal}
Let $\msf A:\Omega_{\msf A}\to\mc L(\hil)$ be an observable. We say that a channel $\Phi:\mc S(\hil)\to\mc S(\hik_\Phi)$ is {\it $\mfr C_\Ao$-universal for $\msf A$} if it is compatible with $\msf A$ and, whenever  $\Lambda$ is a channel compatible with $\Ao$, there is a channel $\Lambda'$ such that $\Lambda=\Lambda'\circ\Phi$.
\end{definition}

According to \cite{HeMi15}, for any observable $\msf A:\Omega_{\msf A}\to\mc L(\hil)$, there is a channel which is $\mfr C_\Ao$-universal for $\msf A$, the least disturbing channel associated with $\msf A$ which shall be described and discussed in depth in Section \ref{sec:leastdisturbing}. Moreover, for this universality one needs the full information preservation power of the least disturbing channel. To clarify the properties of the least disturbing channel, in Section \ref{subsec:char}, we characterize the channels concatenation equivalent with the least disturbing channel.

Since joint measurements with $\msf A$ are a proper subclass of all the sequential processes where $\msf A$ is measured first, the least disturbing channel associated with $\msf A$ is also $\mfr O_\Ao$-universal for $\msf A$. The natural question arises, whether some strictly less information-preserving channel suffices for this seemingly less stringent universality property. In Section  \ref{subsec:extuniversal}, we show that, in the case where $\msf A$ is an extreme observable, the answer to this question is `no'.

To formalize what we mean by information preservation and information yield of a measurement, we need to discuss post-processing relations within the classes of observables and channels.

%%%%%%%%%%%%%%%%%%%%
\subsection{Post-processing of observables and channels}
%%%%%%%%%%%%%%%%%%%%

Measurements of some observables yield more information than those of others. Consider two observables $\Ao:\Omega_{\Ao}\to\mc L(\hil)$ and $\Bo:\Omega_{\Bo}\to\mc L(\hil)$, $\Omega_{\Ao},\,\Omega_{\Bo}\subset\N$. The observable $\Ao$ can be seen as more informative than $\Bo$ if the outcome statistics of $\Bo$ can be classically processed from that of $\Ao$ in a fixed way that does not depend on the initial state of the system being measured. This is formalized in the following definition \cite{MaMu90a}.

\begin{definition}\label{def:postpr}
For the observables $\Ao$ and $\Bo$ introduced above, we denote $\Bo\psleq\Ao$, if there is a stochastic matrix (or Markov matrix) $p=\big(p(k|j)\big)_{k\in\Omega_{\Bo},j\in\Omega_{\Ao}}$, i.e.,\ all the matrix entries are non-negative and $\sum_{k\in\Omega_{\Bo}}p(k|j)=1$ for all $j\in\Omega_{\Ao}$, such that
\begin{equation}
\Bo(k)=\sum_{j\in\Omega_{\Ao}}p(k|j)\Ao(j),\qquad k\in\Omega_{\Bo},
\end{equation}
and say that $\Bo$ is a {\it post-processing} of $\Ao$ and denote $\Bo=\Ao^p$. If also $\Ao\psleq\Bo$, we denote $\Ao\simeq\Bo$ and say that $\Ao$ and $\Bo$ are {\it post-processing equivalent}. We denote the post-processing equivalence class of discrete observables post-processing equivalent with $\Ao$ as $[\Ao]:=\{\Bo\,|\,\Bo\simeq\Ao\}$.
\end{definition}

Post-processing can be defined for channels in a completely analogously to the way it was defined for observables.

\begin{definition}
Let $\hil$, $\mc K_0$, and $\mc K$ be separable Hilbert spaces and $\Lambda_0:\mc S(\hil)\to\mc S(\mc K_0)$ and $\Lambda:\mc S(\hil)\to\mc S(\mc K)$ be channels. We denote $\Lambda_0\psleq\Lambda$ and say that $\Lambda_0$ is a {\it post-processing of $\Lambda$} if there is a channel $\Gamma:\mc S(\mc K)\to\mc S(\mc K_0)$ such that $\Lambda_0=\Gamma\circ\Lambda$. Post-processing equivalence is defined in the obvious way and the equivalence class is denoted $[\Lambda_0]:=\{\Lambda\,|\,\Lambda\simeq\Lambda_0\}$.
\end{definition}

Note that the Schr\"odinger output Hilbert space of $\Lambda\in[\Lambda_0]$ can be any separable Hilbert space. When all the Hilbert spaces considered are separable and isomorphic Hilbert spaces are identified and we only concentrate on discrete observables, all the classes $\mfr O_\Ao$, $\mfr C_\Ao$, $[\Ao]$, and $[\Lambda]$, for an observable $\Ao$ and a channel $\Lambda$, are sets. However, if we do not make this simplification, the subsequent equations involving these classes should be understood as equivalences of classes. This should cause no confusion.

%%%%%%%%%%%%%%%%%%%%%%%
\section{Least disturbing channels}\label{sec:leastdisturbing}
%%%%%%%%%%%%%%%%%%%%%%%

In \cite{HeMi13,HeMi15}, a special maximal compatible channel for each discrete observable $\msf A$ was defined. The original definition of this channel depended on a particular Na\u{\i}mark dilation of $\msf A$. In the following subsection, we give this special definition of this least disturbing channel. In Subsec. \ref{subsec:char}, however, we characterize the full post-processing equivalence class of such a least disturbing channel and show that we do not have to refer to any particular dilation of $\msf A$ in the definition of this equivalence class; indeed, this class only depends on the post-processing equivalence class of $\msf A$.

%%%%%%%%%%%%%%%%%%%%%%%
\subsection{The dilation-dependent form of a least disturbing channel}
%%%%%%%%%%%%%%%%%%%%%%%

In order to introduce the idea of a least disturbing channel for an observable, we need the notion of a Na\u{\i}mark dilation.

\begin{definition}
Let $\msf A:\Omega_{\msf A}\to\mc L(\hil)$ be an observable. We say that a triple $(\mc M,\msf P,J)$ consisting of a Hilbert space $\mc M$, a projection-valued measure $\msf P:\Omega_{\msf A}\to\mc L(\mc M)$, and an isometry $J:\hil\to\mc M$ is a {\it Na\u{\i}mark dilation for $\msf A$} if $\msf A(j)=J^*\msf P(j)J$ for all $j\in\Omega_{\msf A}$. We say that the dilation $(\mc M,\msf P,J)$ is {\it minimal} if the closure of the vector space spanned by $\msf P(j)J\fii$, $j\in\Omega_{\msf A}$, $\fii\in\hil$, is the whole of $\mc M$.
\end{definition}

Every observable has a Na\u{\i}mark dilation and, among these, there is a minimal one \cite{Naimark43}. Suppose that $(\mc M,\msf P,J)$ is a minimal Na\u{\i}mark dilation for an observable $\msf A$ and $(\mc M',\msf P',J')$ is another not necessarily minimal dilation for the same observable. It follows that there is an isometry $W:\mc M\to\mc M'$ such that $W\msf P(j)=\msf P'(j)W$ for all $j\in\Omega_{\msf A}$. Especially, the minimal Na\u{\i}mark dilations for the same observable are mutually unitarily equivalent.

\begin{example}\label{ex:naimark}
In a finite dimensional case we can write a concrete form of a minimal Na\u{\i}mark dilation as follows. We fix a spectral decomposition for each $\Ao(j)$,
\begin{equation}
\Ao(j)=\sum_{k=1}^{r_j}\kb{d_{j,k}}{d_{j,k}} \, , \qquad j\in\Omega_{\Ao} \, .
\end{equation}
We then choose $\mc{M}=\C^{r_1} \oplus \cdots \oplus \C^{r_n}$ and fix an orthonormal basis $\{e_{j,k}\}_{k=1}^{r_j}$ for each $\C^{r_j}$. We define a linear map $J:\hil\to\mc{M}$ as $J\psi = \sum_{j,k} \ip{d_{j,k}}{\psi}e_{j,k}$. Its adjoint $J^*:\mc{M}\to\hil$ is given as $J^*e_{j,k} = d_{j,k}$. The sharp observable that dilates $\Ao$ is given as 
\begin{equation}
\Po(j) = \sum_{k=1}^{r_j}\kb{e_{j,k}}{e_{j,k}} \, , \qquad j\in\Omega_{\Ao} \, .
\end{equation}
The vectors $\Po(j)Jd_{j,k}=e_{j,k}$, $k=1\ldots,\,r_j$, $j\in\Omega_{\Ao}$ span $\mc M$ showing that $(\mc M,\Po,J)$ is a minimal dilation for $\msf A$.
\end{example}

Let $(\mc M,\msf P,J)$ be a minimal Na\u{\i}mark dilation for $\Ao$. The least disturbing channel $\Lambda_\Ao$ of $\Ao$ (associated to the dilation $(\mc M,\msf P,J)$) is defined as
\begin{equation}\label{eq:leastdistr0}
\Lambda_\Ao(\varrho) = \sum_j\Po(j)J\varrho J^*\Po(j) \, , \quad\varrho\in\mc S(\hil).
\end{equation}
The name `least disturbing channel' is justified by the fact that, as shown in \cite{HeMi13}, any channel $\Lambda:\mc S(\hil)\to\mc S(\mc K)$ that is compatible with $\Ao$ is of the form
\begin{equation}
\Lambda = \Lambda'\circ\Lambda_\Ao\, ,
\end{equation}
where $\Lambda':\mc S(\mc M)\to\mc S(\mc K)$ is some channel. Physically this means that $\Lambda$ can be implemented by concatenating $\Lambda'$ and $\Lambda_\Ao$. Hence, whenever one needs to measure $\msf A$ but still wants to have the opportunity to {\it afterwards} carry out some quantum device compatible with $\msf A$ (typically a channel), one can do the following: First measure $\msf A$ such that the associated total state transformation is $\Lambda_{\msf A}$. For any quantum device $\mathcal{D}$ (channel or observable) compatible with $\msf A$, there exists a device $\mathcal{D}'$ on the output system of $\Lambda_{\msf A}$ described by the Hilbert space $\mc M$ such that carrying out $\mathcal{D}'$ after the least disturbing measurement of $\msf A$ gives an implementation of $\mathcal{D}$, i.e.,\ $\mathcal{D}=\mathcal{D}' \circ \Lambda_{\msf A}$.

Let us introduce another (not necessarily minimal) Na\u{\i}mark dilation $(\mc M',\msf P',J')$ for $\Ao$ and define the channel $\Lambda'_{\Ao}:\mc L(\mc M')\to\mc L(\hil)$,
\begin{equation}\label{eq:leastdistr}
\Lambda'_{\msf A}(\varrho)=\sum_{j\in\Omega_{\msf A}}\msf P'(j)J'\varrho(J')^*\msf P'(j),\quad \varrho\in\mc S(\mc \hil).
\end{equation}
We have an isometry $W:\mc M\to\mc M'$ such that $W\msf P(j)=\msf P'(j)W$ for all $j\in\Omega_{\msf A}$ and $WJ=J'$. Define the channels $\Phi:\mc S(\mc M)\to\mc S(\mc M')$ and $\iota:\mc S(\mc M')\to\mc S(\mc M)$,
\begin{equation}
\begin{array}{rcll}
\Phi(\sigma)&=&W\sigma W^*,&\sigma\in\mc S(\mc M)\\
\iota(\sigma')&=&W^*\sigma' W+\tr{(\id_{\mc M}-WW^*)\sigma'}\sigma_0,&\sigma'\in\mc S(\mc M') \, , 
\end{array}
\end{equation}
where $\sigma_0$ is some positive trace-1 operator on $\mc M$. Through direct calculation, one sees that $\Phi\circ\Lambda_{\msf A}=\Lambda'_{\msf A}$ and $\iota\circ\Lambda'_{\msf A}=\Lambda_{\msf A}$. This means that the equivalence class $[\Lambda_{\msf A}]$ does not depend on the choice of the dilation and the dilation does not have to be minimal.

Let $\msf A:\Omega_{\msf A}\to\mc L(\hil)$ and $\Bo:\Omega_{\Bo}\to\mc L(\hil)$ be observables and $\Lambda_{\msf A}$ and $\Lambda_{\Bo}$ be the associated least disturbing channels defined as in \eqref{eq:leastdistr0} with respect to some (minimal) Na\u{\i}mark dilations of $\msf A$ and, respectively, $\Bo$. It was shown in \cite{HeMi13} that $\Bo\psleq\msf A$ if and only if $\Lambda_{\msf A}\psleq\Lambda_{\Bo}$. The physical interpretation of this result is that the more informative observable we measure, the more the measurement disturbs the system. Especially, $\msf A\simeq\Bo$ if and only if $\Lambda_{\msf A}\simeq\Lambda_{\Bo}$.

%%%%%%%%%%%%%%%%%%%%%%%
\subsection{Interlude: minimally sufficient observables}\label{subsec:minsuff}
%%%%%%%%%%%%%%%%%%%%%%%

A special case of post-processing of observables is where $p$ only has entries among $\{0,1\}$. This means that there is a function $f:\Omega_{\Ao}\to\Omega_{\Bo}$ such that $p(k|j)=\delta_{k,f(j)}$, where $\delta_{k,\ell}$ is the Kronecker delta, implying that
\begin{equation}
\Bo(k)=\sum_{j\in f^{-1}(k)}\Ao(j),
\end{equation}
where $f^{-1}(k)$ is the pre-image of $k\in\Omega_{\Bo}$. We then say that $\Bo$ is a {\it relabeling of $\Ao$}, and that $\msf A$ is a {\it refinement} of $\msf B$.

The following definition introduced in \cite{Kuramochi15b} characterizes the observables with minimum informational redundancy.

\begin{definition}\label{def:minsuff}
An observable $\Ao:\Omega_{\Ao}\to\mc L(\hil)$ is {\it minimally sufficient} if, whenever an observable $\Bo$ is post-processing equivalent to $\Ao$, then $\Bo$ is a refinement of $\Ao$.
\end{definition}

A non-vanishing observable $\Ao:\Omega_{\Ao}\to\mc L(\hil)$, i.e.,\ $\Ao(j)\neq0$ for all $j\in\Omega_{\Ao}$, is minimally sufficient if and only if, whenever $i,\,j\in\Omega_{\Ao}$, $i\neq j$, there is no $c>0$ such that $\Ao(i)=c\Ao(j)$. An important fact is that, as shown in \cite{Kuramochi15b}, for any observable $\Ao:\Omega_{\Ao}\to\mc L(\hil)$ (with a separable Hilbert space $\hil$), there is a minimally sufficient representative $\tilde{\Ao}\in[\Ao]$ obtained as a relabeling from $\Ao$, and this minimally sufficient representative is (essentially) unique up to a bijective relabeling of its values. The observable $\tilde{\Ao}$ can be constructed as follows: Define the equivalence relation $\sim$ within $\Omega_{\Ao}$ by declaring $i\sim j$ if there is $c>0$ such that $\Ao(i)=c\Ao(j)$. Denote by $\tilde{\Omega}$ the set $\Omega_{\Ao}/\sim$ of equivalence classes. The observable $\tilde{\Ao}:\tilde{\Omega}\to\mc L(\hil)$ obtained as the following relabeling
\begin{equation}
\tilde{\Ao}([j])=\sum_{i\in[j]}\Ao(i),\qquad [j]\in\tilde{\Omega}
\end{equation}
is post-processing equivalent with $\Ao$ and minimally sufficient. 

There is a useful equivalent characterization for minimal sufficiency for a discrete observable $\Ao:\Omega_{\Ao}\to\mc L(\hil)$ (indeed, for any observable with a standard Borel value space): $\Ao$ is minimally sufficient if and only if $\Ao^p=\Ao$ with some probability matrix $p$ implies $p(i|j)=\delta_{i,j}$ for all $i,\,j\in\Omega_{\Ao}$.

To see this, assume first that $\Ao^p=\Ao$ with some probability matrix $p$ implies $p(i|j)=\delta_{i,j}$ for all $i,\,j\in\Omega_{\Ao}$. Let now $\Bo:\Omega_{\Bo}\to\mc L(\hil)$ be an observable $\Bo\simeq\Ao$, i.e.,\ there are probability matrices $q_1$ and $q_2$ such that $\Bo=\Ao^{q_1}$ and $\Ao=\Bo^{q_2}$. Without loss of generality, we may assume that $\Bo$ is non-vanishing. Defining the probability matrix $p(\cdot|\cdot):\Omega_{\Ao}\times\Omega_{\Ao}\to[0,1]$, $p(i|j)=\sum_{r\in\Omega_{\Bo}}q_2(i|r)q_1(r|j)$, $i,\,j\in\Omega_{\Ao}$, we have $\Ao^p=\Ao$. Thus, $p(i|j)=\delta_{i,j}$ for all $i,\,j\in\Omega_{\Ao}$. Using Cauchy-Schwartz inequality, we have for every $i,\,j\in\Omega_{\Ao}$
\begin{equation}
\begin{array}{r}
\delta_{i,j}   =\sum_{k\in\Omega_{\Ao}}\delta_{k,i}\delta_{k,j}=\sum_{k\in\Omega_{\Ao}}\sqrt{\sum_{r,s\in\Omega_{\Bo}}q_2(k|r)q_1(r|i)q_2(k|s)q_1(s|j)}\\
 \geq \sum_{k\in\Omega_{\Ao}}\sum_{r\in\Omega_{\Bo}}q_2(k|r)\sqrt{q_1(r|i)q_2(r|j)}=\sum_{r\in\Omega_{\Bo}}\sqrt{q_1(r|i)q_2(r|j)}
\end{array}
\end{equation}
This implies that whenever $i\neq j$, $q_1(r|i)q_1(r|j)=0$ for all $r\in\Omega_{\Bo}$. This, in turn, means that there is a partition $\Omega_{\Bo}=\bigcup_{j\in\Omega_{\Ao}}\Omega_j$ such that $p(r|j)\neq0$ only if $r\in\Omega_j$. Using the fact that $\Bo$ is non-vanishing, one immediately sees that $p(r|j)\neq0$ if and only if $r\in\Omega_j$.

Suppose that there are $i\in\Omega_{\Ao}$ and $r\in\Omega_{\Bo}$ such that $0<q_2(i|r)<1$. Let $j\in\Omega_{\Ao}$ be the outcome such that $r\in\Omega_j$. There are two possibilities: $i=j$ or $i\neq j$. If $i=j$, using $q_2(i|r)<1$, one obtains
\begin{equation}
1=\sum_{s\in\Omega_{\Bo}}q_2(i|s)q_1(s|i)<\sum_{s\in\Omega_j}q_1(s|i)=1 \, , 
\end{equation}
a contradiction. 
If $i\neq j$, we get
\begin{equation}
0=\sum_{s\in\Omega_{\Bo}}q_2(i|s)q_1(s|i)\geq q_2(i|r)q_1(r|j) \, , 
\end{equation}
which, since $q_2(i|r)\neq0$, implies $q_1(r|j)=0$. However, this is impossible, since $r\in\Omega_j$. Thus, $q_2(i|r)\in\{0,1\}$ for all $i\in\Omega_{\Ao}$ and all $r\in\Omega_{\Bo}$, meaning that $q_2$ is a coarse-graining. This means that $\Ao$ is minimally sufficient.

Suppose now that $\Ao$ is minimally sufficient and let $p$ be a probability matrix such that $\Ao^p=\Ao$. Without loss of generality, we may assume that $\Ao$ is non-vanishing. Define the observable $\Bo:\Omega_{\Ao}\times\Omega_{\Ao}\to\mc L(\hil)$,
\begin{equation}
\Bo(i,j)=p(i|j)\Ao(j) \, ,\qquad i,\,j\in\Omega_{\Ao} \, .
\end{equation}
Clearly, $\Bo\simeq\Ao$. Define the function $f:\Omega_{\Ao}\times\Omega_{\Ao}\to\Omega_{\Ao}$, $f(i,j)=i$ for all $i,\,j\in\Omega_{\Ao}$. It follows that
\begin{equation}
\sum_{(k,j)\in f^{-1}(i)}\Bo(k,j)=\sum_{j\in\Omega_{\Ao}}p(i|j)\Ao(j)=\Ao(i),\quad i\in\Omega_{\Ao} \, , 
\end{equation}
i.e.,\ the relabeling mediated by $f$ takes $\Bo$ into $\Ao\simeq\Bo$. In statistical terms, this means that this coarse-graining (statistic) is {\it sufficient} for $\Bo$ and, according to \cite[Prop. 5]{Kuramochi15b}, there is a function $h:\Omega_{\Ao}\times\Omega_{\Ao}\to[0,\infty)$ and a positive-operator-valued function $G:\Omega_{\Ao}\to\mc L(\hil)$ such that
\begin{equation}\label{eq:suffstat}
p(i|j)\Ao(j)=\Bo(i,j)=h(i,j)(G\circ f)(i,j)=h(i,j)G(i),\quad i,\,j\in\Omega_{\Ao}.
\end{equation}
Summing over $j$ in \eqref{eq:suffstat}, and denoting $H(i)=\sum_jh(i,j)$, we obtain $H(i)G(i)=\Ao(i)$ for all $i\in\Omega_{\Ao}$. Since $\Ao$ is non-vanishing, $H(i)\neq0$ for all $i\in\Omega_{\Ao}$ and, substituting this to \eqref{eq:suffstat}, we get
\begin{equation}
\frac{h(i,j)}{H(i)}\Ao(i)=p(i|j)\Ao(j),\qquad i,\,j\in\Omega_{\Ao}.
\end{equation}
Make now the counter assumption that there are $i,\,j\in\Omega_{\Ao}$, $i\neq j$, such that $p(i|j)\neq0$. 
Now $h(i,j)\neq0$, for otherwise \eqref{eq:suffstat} would imply $\Ao(j)=0$. 
Hence,
\begin{equation}
\Ao(i)=p(i|j)\frac{H(i)}{h(i,j)}\Ao(j),
\end{equation}
where the factor before $\Ao(j)$ is positive. According to our earlier characterization for minimal sufficiency of a discrete POVM, this is impossible. It follows that $p(i|j)=\delta_{i,j}$ for all $i,\,j\in\Omega_{\Ao}$.

%%%%%%%%%%%%%%%%%%%%%%%
\subsection{Characterization of least disturbing channels}\label{subsec:char}
%%%%%%%%%%%%%%%%%%%%%%%

Any representative of the equivalence class $[\Lambda_{\msf A}]$ also deserves to be called as least disturbing, and this is what we will subsequently do. 

Recall the definition of $\mfr C_\Ao$-universality (Definition \ref{def:chuniversal}). As has already been mentioned, the channels $\Lambda_{\msf A}$ as defined in \eqref{eq:leastdistr0} and \eqref{eq:leastdistr} are $\mfr C_\Ao$-universal for $\msf A$. Hence all the least disturbing channels have this universality property. However, the converse is also true: Suppose that $\Phi$ is $\mfr C_\Ao$-universal for $\msf A$. Since $\Phi$ is compatible with $\msf A$, $\Phi\psleq\Lambda_{\msf A}$ with some $\Lambda_{\msf A}$ as defined in \eqref{eq:leastdistr0}. Moreover, since this $\Lambda_{\msf A}$ is compatible with $\msf A$, the definition of $\mfr C_\Ao$-universality implies that $\Lambda_{\msf A}\psleq\Phi$. Hence $\Phi\in[\Lambda_{\msf A}]$. We find that the class $[\Lambda_{\msf A}]$ of least disturbing channels for $\msf A$ can equivalently be characterized as the set of $\mfr C_\Ao$-universal channels for $\msf A$.

In this subsection, we characterize the least disturbing channels associated to a discrete observable. Let us first state and prove a simple useful lemma.

\begin{lemma}\label{lemma:suppLambda}
Suppose that $\msf A:\Omega_{\msf A}\to\mc L(\hil)$ is an observable and $(\mc M,\msf P,J)$ is a minimal Na\u{\i}mark dilation for $\msf A$. Let $\Lambda_{\msf A}:\mc S(\hil)\to\mc S(\mc M)$ be the least disturbing channel,
\begin{equation}
\Lambda_{\msf A}(\varrho)=\sum_{j\in\Omega_{\msf A}}\msf P(j)J\varrho J^*\msf P(j),\qquad\varrho\in\mc S(\hil).
\end{equation}
Whenever $B\in\mc L(\mc M)$ is positive, $\Lambda_{\msf A}^*(B)=0$ implies $B=0$.
\end{lemma}

\begin{proof}
Let $B\in\mc L(\mc M)$ be positive and $\Lambda_{\msf A}^*(B)=0$. We find for all $\fii\in\hil$
\begin{equation}
0=\<\fii|\Lambda_{\msf A}^*(B)\fii\>=\sum_{j\in\Omega_{\msf A}}\<\msf P(j)J\fii|B\msf P(j)J\fii\>
\end{equation}
and, since all the summands are non-negative and vectors $\msf P(j)J\fii$, $\fii\in\hil$ span $\msf P(j)\mc M$ for each $j\in\Omega_{\msf A}$, we have $\msf P(j)B\msf P(j)=0$ for all $j\in\Omega_{\msf A}$.

Let $\{e_n\}_n\subset\mc M$ be an orthonormal basis diagonalizing $\msf P$. Especially, $\<e_n|Be_n\>=0$ for all $n$. For any $m$ and $n$,
\begin{equation}
0\leq
\left(\begin{array}{cc}
\<e_m|Be_m\>&\<e_m|Be_n\>\\
\<e_n|Be_m\>&\<e_n|Be_n\>
\end{array}
\right)=
\left(\begin{array}{cc}
0&\<e_m|Be_n\>\\
\<e_n|Be_m\>&0
\end{array}
\right) \, , 
\end{equation}
which is possible if and only if $\<e_m|Be_n\>=0$. Since this holds for any $m$ and $n$, $B=0$.
\end{proof}

Especially the above lemma implies that, whenever $\Lambda_{\msf A}$ is obtained from a minimal Na\u{\i}mark dilation $(\mc M,\msf P,J)$ of $\msf A$, the support projection (see the appendix, Lemma \ref{lemma:suppPhi-Sapfo}) of $\Lambda_{\msf A}$ is $\id_{\mc M}$. Indeed, if $R\in\mc L(\mc M)$ is a projection such that $\Lambda_{\msf A}^*(R)=\id_\hil$, then $\Lambda_{\msf A}^*(\id_{\mc M}-R)=0$. According to Lemma \ref{lemma:suppLambda}, this means that $R=\id_{\mc M}$.

\begin{proposition}\label{prop:MinLeast}
Let $\msf A:\Omega_{\msf A}\to\mc L(\hil)$ be a minimally sufficient observable and $(\mc M,\msf P,J)$ a minimal Na\u{\i}mark dilation for $\msf A$. Fix the least disturbing channel $\Lambda_{\msf A}:\mc S(\hil)\to\mc S(\mc M)$ associated with this dilation. For any channel $\Gamma:\mc S(\mc M)\to\mc S(\mc M)$ such that $\Gamma\circ\Lambda_{\msf A}=\Lambda_{\msf A}$, one has $\Gamma\circ\mathbb E_{\msf P}=\mathbb E_{\msf P}$, where $\mathbb E_{\msf P}:\mc S(\mc M)\to\mc S(\mc M)$ is the L\"uders channel,
\begin{equation}
\mathbb E_{\msf P}(\varrho)=\sum_{j\in\Omega_{\msf A}}\msf P(j)\varrho\msf P(j) \, , \qquad\varrho\in\mc L(\mc M) \, .
\end{equation}
\end{proposition}

\begin{proof}
For simplicity, we work in the dual (Heisenberg picture) in this proof. Let $\Gamma$ be a channel as in the claim. Denote $\tilde{\Gamma}=\Gamma\circ\mathbb E_{\msf P}$. Note that the range of $\tilde{\Gamma}^*$ is contained in the subalgebra $\mc D\subset\mc L(\mc M)$ of those operators $D$ commuting with $\msf P$. Let $F\in\mc L(\mc M)$ be such that $\tilde{\Gamma}^*(F)=F$; obviously $F\in\mc D$. Using the Schwartz inequality,
\begin{equation}
\tilde{\Gamma}^*(F^*F)\geq\tilde{\Gamma}^*(F)^*\tilde{\Gamma}^*(F)=F^*F \, .
\end{equation}
Now, 
\begin{equation}
\Lambda_{\msf A}^*\big(\tilde{\Gamma}^*(F^*F)-F^*F\big)=\Lambda_{\msf A}^*(F^*F)-\Lambda_{\msf A}^*(F^*F)=0
\end{equation}
implying, since $\tilde{\Gamma}^*(F^*F)-F^*F\geq0$, that $\tilde{\Gamma}^*(F^*F)=F^*F$ using Lemma \ref{lemma:suppLambda}. Hence the set of fixed points for $\tilde{\Gamma}^*$ is a subalgebra of the multiplicative domain of $\tilde{\Gamma}^*$, the von Neumann algebra consisting of those $B\in\mc L(\mc M)$ such that $\tilde{\Gamma}^*(BC)=\tilde{\Gamma}^*(B)\tilde{\Gamma}^*(C)$ and $\tilde{\Gamma}^*(CB)=\tilde{\Gamma}^*(C)\tilde{\Gamma}^*(B)$ for all $C\in\mc L(\mc M)$. Let us denote the $\sigma$-weak closure of the fixed point set by $\mc F$; this is a von Neumann algebra. Let $\mathbb E^*:\mc L(\mc M)\to\mc F$ be the conditional expectation whose existence is guaranteed, e.g.,\ by \cite{KuNa79}. Thus $\mathbb E$ is a normal completely positive unital map such that $\mathbb E^*(F_1BF_2)=F_1\mathbb E^*(B)F_2$ for all $F_1,\,F_2\in\mc F$ and $B\in\mc L(\mc M)$ and $\mathbb E^*\circ\tilde{\Gamma}^*=\tilde{\Gamma}^*\circ\mathbb E^*=\mathbb E^*$. Moreover, $\Lambda_{\msf A}^*\circ\mathbb E^*=\Lambda_{\msf A}^*$. In fact, in the Schr\"odinger picture,
\begin{equation}
\mathbb E=\lim_{N\to\infty}\frac{1}{N}\sum_{n=0}^{N-1}\tilde{\Gamma}^n,
\end{equation}
where $\tilde{\Gamma}^0$ is the identity map and the limit is with respect to the operator norm for trace-norm-continuous mappings within the trace class of $\mc M$. In what follows, we prove that the centres of $\mc F$ and $\mc D$ coincide, and using this result, we easily see that, in fact, $\mathbb E=\mathbb E_{\msf P}$. Thus, we first study the structure of the centre of $\mc F$. We first characterize the central projections of $\mc F$ and then use the spectral theorem to characterize the whole centre.

Denote the centre $\mc F\cap\mc F'=:\mc Z(\mc F)$ and let $Q\in\mc Z(\mc F)\subset\mc D$ be a projection. Define the set $\Omega_{\msf B}:=\Omega_{\msf A}\times\{0,1\}$, the projection-valued measure $\msf Q:\Omega_{\msf B}\to\mc L(\mc M)$,
\begin{equation}
\msf Q(j,0)=\msf P(j)Q,\quad\msf Q(j,1)=\msf P(j)(\id_{\mc M}-Q),\qquad j\in\Omega_{\msf A} \, , 
\end{equation}
and the observable $\msf B:\Omega_{\msf B}\to\mc L(\hil)$, $\msf B=\Lambda_{\msf A}^*\circ\msf Q$. Clearly, $(\mc M,\msf Q,J)$ is a minimal Na\u{\i}mark dilation for $\msf B$. Let $\Lambda_{\msf B}$ be the least disturbing channel for $\msf B$ defined by this dilation, i.e.,\ $\Lambda_{\msf B}^*(B)=\sum_{j\in\Omega_{\msf A}}J^*\msf P(j)\big(QBQ+(\id_{\mc M}-Q)B(\id_{\mc M}-Q)\big)\msf P(j)J$. Using the properties of the conditional expectation $\mathbb E$, we find
\begin{eqnarray*}
\Lambda_{\msf B}^*(B)&=&\sum_{j\in\Omega_{\msf A}}J^*\msf P(j)\big(QBQ+(\id_{\mc M}-Q)B(\id_{\mc M}-Q)\big)\msf P(j)J\\
&=&\Lambda_{\msf A}^*\big(QBQ+(\id_{\mc M}-Q)B(\id_{\mc M}-Q)\big)\\
&=&(\Lambda_{\msf A}^*\circ\mathbb E^*)\big(QBQ+(\id_{\mc M}-Q)B(\id_{\mc M}-Q)\big)\\
&=&\Lambda_{\msf A}^*\big(Q\mathbb E^*(B)Q+(\id_{\mc M}-Q)\mathbb E^*(B)(\id_{\mc M}-Q)\big)\\
&=&(\Lambda_{\msf A}^*\circ\mathbb E^*)(B)=\Lambda_{\msf A}^*(B).
\end{eqnarray*}
In the second to last equality, e.g.,\ we have used the fact that $Q\in\mc Z(\mc F)$ and the fact that the range of $\mathbb E^*$ is $\mc F$. Thus, especially, $\msf A\simeq\msf B$.

Next we show that $Q$ is a sum of the projections $\msf P(j)$. In order to do this, the minimal sufficiency of $\msf A$ is used. Let $p=\big(p(i,r|j)\big)_{i,j\in\Omega_{\msf A},\ r\in\{0,1\}}$ be a probability matrix such that $\msf B=\msf A^p$. Thus, defining $q(i|j)=p(i,0|j)+p(i,1|j)$ for all $i,\,j\in\Omega_{\msf A}$, we have $\msf A^q=\msf A$, and since $\msf A$ is minimally sufficient, $q(i|j)=\delta_{i,j}$ for all $i,\,j\in\Omega_{\msf A}$. If, for some $i,\,j\in\Omega_{\msf A}$, $0<p(i,0|j)<1$, then automatically $i=j$, and there must be $k\neq i$ and $r\in\{0,1\}$ such that $p(k,r|i)>0$. Thus, denoting $s=2^{-1}\big(1+(-1)^r\big)$,
\begin{equation}
0=\delta_{k,i}=p(k,r|i)+p(k,s|i)>0 \, , 
\end{equation}
a contradiction. Thus, for all $i,\,j\in\Omega_{\msf A}$ and $r\in\{0,1\}$, $p(i,r|j)\in\{0,1\}$, implying that there is a function $\chi:\Omega_{\msf B}\to\{0,1\}$ such that $\msf B(i,r)=\chi(i,r)\msf A(i)$ for all $(i,r)\in\Omega_{\msf A}$. Pick $j\in\Omega_{\msf A}$ and $\fii\in\hil$ and denote $\eta:=\msf P(j)J\fii$. Now,
\begin{equation}
\<\eta|\msf P(j)Q\eta\>=\<\fii|\msf B(j,0)\fii\>=\chi(j,0)\<\eta|\eta\>
\end{equation}
and, since vectors like $\eta$ span $\msf P(j)\mc M$, we have $\msf P(j)Q=\chi(j,0)\msf P(j)$. Thus $Q=\sum_{j\in\Omega_{\msf A}}\msf P(j)Q=\sum_{j\in\Omega_{\msf A}}\chi(j,0)\msf P(j)$: all projections in $\mc Z(\mc F)$ are sums of $\msf P(j)$.

Let $E\in\mc Z(\mc F)$, $0\leq E\leq\id_{\mc M}$. According to the above and the spectral theorem, there are $e_j\in[0,1]$, $j\in\Omega_{\msf A}$, such that $E=\sum_{j\in\Omega_{\msf A}}e_j\msf P(j)$. Especially, there are $e(i|j)\in[0,1]$, $i,\,j\in\Omega_{\msf A}$ such that $\mathbb E^*\big(\msf P(i)\big)=\sum_{j\in\Omega_{\msf A}}e(i|j)\msf P(j)$ for all $i\in\Omega_{\msf A}$. It is simply checked that $e=\big(e(i|j)\big)_{i,\,j\in\Omega_{\msf A}}$ is a probability matrix and $\msf A^e=\msf A$. Thus, again, $e(i|j)=\delta_{i,j}$ for all $i,\,j\in\Omega_{\msf A}$ implying $\mathbb E^*\circ\msf P=\msf P$. This means that $\mc F$ and $\mc D$ have a common centre consisting of all the complex linear combinations of $\msf P(j)$, $j\in\Omega_{\msf A}$.

Next, we show that $\mathbb E=\mathbb E_{\msf P}$ which completes the proof. Assume that $B\in\mc L\big(\msf P(j)\mc M\big)$ for some $j\in\Omega_{\msf A}$. Note that, since $\mathbb E^*(B)^*\mathbb E^*(B)\leq\mathbb E^*(B^*B)\leq\|B\|^2\mathbb E^*\big(\msf P(j)\big)=\|B\|^2\msf P(j)$, one has $\mathbb E^*(B)\in\mc L\big(\msf P(j)\mc M\big)$. Let $\fii\in\hil$ and denote $\eta=\msf P(j)J\fii$. We have
\begin{eqnarray*}
\<\eta|B\eta\>&=&\<\fii|\Lambda_{\msf A}^*(B)\fii\>=\<\fii|(\Lambda_{\msf A}^*\circ\mathbb E^*)(B)\fii\>\\
&=&\<\fii|(\Lambda_{\msf A}^*\circ\mathbb E^*)\big(\msf P(j)B\msf P(j)\big)\fii\>\\
&=&\<\fii|\Lambda_{\msf A}^*\big(\msf P(j)\mathbb E^*(B)\msf P(j)\big)\fii\>=\<\eta|\mathbb E^*(B)\eta\>.
\end{eqnarray*}
Since vectors like $\eta$ span $\msf P(j)\mc M$, we have $\mathbb E^*(B)=B$. Thus we obtain for any $B\in\mc L(\mc M)$
\begin{eqnarray*}
\mathbb E_{\msf P}^*(B)&=&\sum_{j\in\Omega_{\msf A}}\msf P(j)B\msf P(j)=\sum_{j\in\Omega_{\msf A}}\mathbb E^*\big(\msf P(j)B\msf P(j)\big)\\
&=&\sum_{j\in\Omega_{\msf A}}\msf P(j)\mathbb E^*(B)\msf P(j)=\sum_{j\in\Omega_{\msf A}}\msf P(j)\mathbb E^*(B)=\mathbb E^*(B).
\end{eqnarray*}
Thus $\mathbb E=\mathbb E_{\msf P}$ and $\mathbb E_{\msf P}^*\circ\Gamma^*=\mathbb E_{\msf P}^*\circ\tilde{\Gamma}^*=\mathbb E^*\circ\tilde{\Gamma}^*=\mathbb E^*=\mathbb E_{\msf P}^*$.
\end{proof}

Any channel $\Lambda:\mc S(\hil)\to\mc S(\mc K)$ compatible with an observable $\msf A$ has the form $\Lambda=\Gamma\circ\Lambda_{\msf A}$ with some channel $\Gamma$ and $\Lambda_{\msf A}$ defined as in \eqref{eq:leastdistr0} with a minimal Na\u{\i}mark dilation $(\mc M,\msf P,J)$ of $\msf A$. Let the channel $\mathbb E_{\msf P}$ be as in Proposition \ref{prop:MinLeast}. Since $\Lambda_{\msf A}=\mathbb E_{\msf P}\circ\Lambda_{\msf A}$, we have $\Lambda=\tilde{\Gamma}\circ\Lambda_{\msf A}$, where $\tilde{\Gamma}:=\Gamma\circ\mathbb E_{\msf P}$. Denoting, for each $j\in\Omega_\Ao$, $\Gamma_j:=\Gamma|_{\mc S\big(\msf P(j)\mc M\big)}$, we have
\begin{equation}\label{eq:compch}
\Lambda(\varrho)=\sum_{j\in\Omega_{\msf A}}\Gamma_j\big(\Po(j)J\varrho J^*\Po(j)\big),\quad \varrho\in\mc S(\hil).
\end{equation}

\begin{corollary}\label{lemma:fiberwise_id}
Suppose that $\msf A$ and $\Lambda_{\msf A}$ are as in the claim of Proposition \ref{prop:MinLeast}. Let $\Lambda\in[\Lambda_{\msf A}]$, $\Lambda:\mc S(\hil)\to\mc S(\mc K)$, $\mc K$ a Hilbert space, and let $\Gamma_j:\mc S\big(\msf P(j)\mc M\big)\to\mc S(\mc K)$, $j\in\Omega_{\msf A}$, be channels such that \eqref{eq:compch} holds. For each $j\in\Omega_{\msf A}$, $\Gamma_j\simeq\mathrm{id}_{\mc S\big(\msf P(j)\mc M\big)}$.
\end{corollary}

\begin{proof}
Define channel $\Gamma:\mc S(\mc M)\to\mc S(\mc K)$ as
\begin{equation}
\Gamma(\sigma)=\sum_{j\in\Omega_{\msf A}}\Gamma_j\big(\Po(j)\sigma\Po(j)\big),\qquad \sigma\in\mc S(\mc M) \, , 
\end{equation}
so that $\Lambda=\Gamma\circ\Lambda_{\msf A}$. 
Let $\Gamma':\mc S(\mc K)\to\mc S(\mc M)$ be a channel such that $\Lambda_{\msf A}=\Gamma'\circ\Lambda$. Proposition \ref{prop:MinLeast} implies that $\Gamma'\circ\Gamma\circ\mathbb E_{\msf P}=\mathbb E_{\msf P}$, and thus, for all $j\in\Omega_{\msf A}$ and all $\sigma\in\mc S(\mc M)$, $(\Gamma\circ\Gamma_j)(\sigma)=\msf P(j)\sigma\msf P(j)$. Defining $\Gamma'_j:=\Po(j)\Gamma'(\cdot)\Po(j)$ for all $j\in\Omega_{\msf A}$, we get $\Gamma'_j\circ\Gamma_j=\mathrm{id}_{\mc S\big(\msf P(j)\mc M\big)}$.
\end{proof}

To clarify the equivalence class $[\Lambda_{\msf A}]$, we need to define the notion of Stinespring dilations and discuss their basic properties.

\begin{definition}
Suppose that $\mc K$ is another Hilbert space and $\Lambda:\mc S(\hil)\to\mc S(\mc K)$ is a channel. A pair $(\mc N,V)$ consisting of a Hilbert space $\mc N$ and an isometry $V:\hil\to\mc K\otimes\mc N$ is a {\it Stinespring dilation for $\Lambda$} if $\Lambda^*(B)=V^*(B\otimes\id_{\mc N})V$ for all $B\in\mc L(\mc K)$. This dilation is {\it minimal} if the vectors $(B\otimes\id_{\mc N})V\fii$, $B\in\mc L(\mc K)$, $\fii\in\hil$, span a dense subspace of $\mc K\otimes\mc N$.
\end{definition}

Every channel has a Stinespring dilation and among the dilations there is always a minimal one which is unique up to unitary equivalence \cite{Stinespring55}. Suppose that $(\mc N,V)$ is a minimal Stinespring dilation for a channel $\Lambda:\mc S(\hil)\to\mc S(\mc K)$ and $(\mc N',V')$ is another not necessarily minimal Stinespring dilation for $\Lambda$. There is an isometry $W:\mc N\to\mc N'$ such that $V'=(\id_{\mc K}\otimes W)V$.

The next theorem states that, for any observable $\msf A$, the class $[\Lambda_{\msf A}]$ of least disturbing channels consists just of the channels $\Lambda_{\msf B}$ with observables $\msf B$ whose measurements yield the same classical information as those of $\msf A$.

\begin{theorem}\label{theor:leastdistreq}
Let $\msf A:\Omega_{\msf A}\to\mc L(\hil)$ be an observable, $\Omega_{\msf A}\subset\N$, and $\Lambda_{\msf A}$ be as in \eqref{eq:leastdistr0}. The class of least disturbing channels is
$$
[\Lambda_{\msf A}]=\{\Lambda_{\Bo}\,|\,\Bo\ \mr{an\ observable,}\ \Bo\simeq\msf A\}.
$$
\end{theorem}

\begin{proof}
It suffices to prove the claim in the case where $\msf A$ is minimally sufficient. Namely, if this was not the case, we would find a minimally sufficient observable $\tilde{\msf A}\simeq\msf A$ with a channel $\Lambda_{\tilde{\msf A}}$ defined by some dilation of $\tilde{\msf A}$ analogously to \eqref{eq:leastdistr0} and \eqref{eq:leastdistr}. Since $\tilde{\msf A}\simeq\msf A$, we have $\Lambda_{\tilde{\msf A}}\simeq\Lambda_{\msf A}$, implying
\begin{equation}
[\Lambda_{\msf A}]=\{\Lambda\,|\,\Lambda\simeq\Lambda_{\msf A}\}=\{\Lambda\,|\,\Lambda\simeq\Lambda_{\tilde{\msf A}}\}=[\Lambda_{\tilde{\msf A}}] \, .
\end{equation}

Let us assume that $\msf A$ is minimally sufficient and $\Lambda:\mc S(\hil)\to\mc S(\mc K)$ is a channel, $\Lambda\in[\Lambda_{\msf A}]$. Suppose that $\Lambda_{\msf A}$ is defined by a minimal Na\u{\i}mark dilation $(\mc M,\msf P,J)$ as in \eqref{eq:leastdistr0} and denote $\msf P(j)\mc M=\mc M_j$ for all $j\in\Omega_{\msf A}$. Let the channels $\Gamma_j:\mc S(\mc M_j)\to\mc S(\mc K)$ be as in \eqref{eq:compch}.

There is a Hilbert space $\mc N$ such that, for any $j\in\Omega_{\msf A}$, the channel $\Gamma_j$ has a (not necessarily minimal) dilation $(\mc N,V_j)$. Define $\ell^2(\Omega_{\msf A})=:\mc K_{\msf A}$ with the canonical orthonormal basis $\{e_j\}_{j\in\Omega_{\msf A}}$. One may construct the Stinespring dilation $(\mc N\otimes\mc K_{\msf A},V)$ for $\Lambda$,
\begin{equation}
V\fii=\sum_{j\in\Omega_{\msf A}}V_j\msf P(j)J\fii\otimes e_j,\qquad\fii\in\hil \, .
\end{equation}
Since $\Lambda_{\msf A}\psleq\Lambda$, there is a channel $\Phi:\mc L(\mc M)\to\mc L(\mc K)$ such that $\Lambda_{\msf A}=\Lambda\circ\Phi$. Let $(\mc N_\Phi,V_\Phi)$ be a dilation of $\Phi$, and define the minimal dilation $(\mc K_{\msf A},V_{\msf A})$ of $\Lambda_{\msf A}$,
\begin{equation}
V_{\msf A}\fii=\sum_{j\in\Omega_{\msf A}}\msf P(j)J\fii\otimes e_j \, .
\end{equation}
It follows that $\big(\mc N_\Phi\otimes\mc N\otimes\mc K_{\msf A},(V_\Phi\otimes\id_{\mc N}\otimes\id_{\mc K_{\msf A}})V\big)$ is a Stinespring dilation for $\Lambda_{\msf A}$ implying that there is an isometry $W:\mc K_{\msf A}\to\mc N_\Phi\otimes\mc N\otimes\mc K_{\msf A}$ such that
\begin{equation}\label{eq:intertwine}
(\id_{\mc M}\otimes W)V_{\msf A}=(V_\Phi\otimes\id_{\mc N}\otimes\id_{\mc K_{\msf A}})V.
\end{equation}

For each $j\in\Omega_{\msf A}$, there is a sequence of vectors $\eta_{j,k}\in\mc N_\Phi\otimes\mc N$, $k\in\Omega_{\msf A}$, such that $We_j=\sum_k\eta_{j,k}\otimes e_k$. Since $W$ is an isometry, one finds that
\begin{equation}\label{eq:orth}
\sum_{k\in\Omega_{\msf A}}\<\eta_{i,k}|\eta_{j,k}\>=\delta_{i,j},\qquad i,\,j\in\Omega_{\msf A}.
\end{equation}
Equation \eqref{eq:intertwine} implies that
\begin{equation}
(V_\Phi\otimes\id_{\mc N})V_k\msf P(k)J\fii=\sum_{i\in\Omega_{\msf A}}\msf P(i)J\fii\otimes\eta_{i,k},\quad k\in\Omega_{\msf A}, \fii\in\hil \, .
\end{equation}
Taking the norm squared of the above equation, one obtains
\begin{equation}
\msf A(k)=\sum_{j\in\Omega_{\msf A}}\|\eta_{j,k}\|^2\msf A(j),\quad k\in\Omega_{\msf A} \, .
\end{equation}
Since $\msf A$ is minimally sufficient, this together with \eqref{eq:orth} means that $\eta_{j,k}=\delta_{j,k}\eta_j$, where $\{\eta_j\}_{j\in\Omega_{\msf A}}\subset\mc N_\Phi\otimes\mc N$ is an orthonormal system such that $We_j=\eta_j\otimes e_j$ for all $j\in\Omega_{\msf A}$. Equation \eqref{eq:intertwine} now implies
\begin{equation}
\msf P(j)J\fii\otimes\eta_j=(V_\Phi\otimes\id_{\mc N})V_j\msf P(j)J\fii,\qquad\fii\in\hil \, , 
\end{equation}
which can be used to show that, whenever $i\neq j$ and $F\in\mc L(\mc N)$,
\begin{equation}\label{eq:nullify}
\msf P(i)V_i^*(\id_{\mc M}\otimes F)V_j\msf P(j)=0.
\end{equation}

According to Corollary \ref{lemma:fiberwise_id}, $\Gamma_j\simeq\mr{id}_{\mc S(\mc M_j)}$ for all $j\in\Omega_{\msf A}$ which, in turn according to Lemma \ref{lemma:eqid}, is equivalent, for each $j\in\Omega_{\msf A}$, with the existence of index sets $X_j$, probability distributions $\{p_{j,n}\}_{n\in X_j}$, isometries $\{V_{j,n}:\mc M_j\to\mc K\}_{n\in X_j}$, and orthonormal systems $\{b_{j,n}\}_{n\in X_j}\subset\mc N_0$ where the Hilbert space $\mc N_0$ is fixed, such that $V_{jn}^*V_{jm}=\delta_{n,m}\msf P(j)$ and
\begin{equation}
V_j\xi=\sum_{n\in X_j}\sqrt{p_{j,n}}V_{j,n}\xi\otimes b_{j,n},\quad\xi\in\mc M_j.
\end{equation}
With respect to our earlier notations, we may choose $\mc N_0=\mc N$. 
Fix $i,\,j\in\Omega_{\msf A}$, $i\neq j$, $n\in X_i$, and $m\in X_j$, and set $F=|b_{i,n}\>\<b_{j,m}|$ in \eqref{eq:nullify}. 
This yields
\begin{equation}
\sqrt{p_{i,n}p_{j,m}}\<V_{i,n}\msf P(i)\psi|V_{j,m}\msf P(j)\psi\>=0,\qquad\psi\in\mc M \, .
\end{equation}

The above means that, defining the set $\Omega_{\Bo}:=\{(j,n)\,|\,n\in X_j,\ j\in\Omega_{\msf A}\}$, we may define the PVM $\msf Q:\Omega_{\Bo}\to\mc L(\mc K)$ where, for each $j\in\Omega_{\msf A}$ and $n\in X_j$, $\msf Q(j,n)$ is the orthogonal projection onto the range of $V_{j,n}\msf P(j)$. It is simple to check that
\begin{equation}
K:=\sum_{j\in\Omega_{\msf A}}\sum_{n\in X_j}\sqrt{p_{j,n}}V_{j,n}\msf P(j)J
\end{equation}
is an isometry and
\begin{equation}
\Lambda(\varrho)=\sum_{j\in\Omega_{\msf A}}\sum_{n\in X_j}\msf Q(j,n)K\varrho K^*\msf Q(j,n),\qquad\varrho\in\mc S(\hil) \, .
\end{equation}
It thus follows that, defining the observable $\Bo:\Omega_{\Bo}\to\mc L(\hil)$, $\Bo(j,n)=K^*\msf Q(j,n)K$, $\Lambda$ is a least disturbing channel associated to $\Bo$ defined by the Na\u{\i}mark dilation $(\mc K,\msf Q,K)$ of $\Bo$. Moreover, it follows immediately that $\Bo(j,n)=p_{j,n}\msf A(j)$ for all $(j,n)\in\Omega_{\Bo}$, so that $\msf A\simeq\Bo$.

It remains to be shown that, whenever $\Lambda=\Lambda_{\Bo}$ is a channel defined by a dilation of some observable $\Bo\simeq\msf A$ as in \eqref{eq:leastdistr}, then $\Lambda\simeq\Lambda_{\msf A}$. This is, however, immediate since, according to \cite{HeMi13}, $\Bo\simeq\msf A$ is equivalent with $\Lambda_{\Bo}\simeq\Lambda_{\msf A}$.
\end{proof}

%%%%%%%%%%%%%%%%%%
\subsection{Minimal output dimension }
%%%%%%%%%%%%%%%%%

When we fix an observable $\msf A:\Omega_{\msf A}\to\mc L(\hil)$ and define the least disturbing channel $\Lambda_{\msf A}$ with respect to a minimal Na\u{\i}mark dilation $(\mc M,\msf P,J)$ of $\msf A$ as in \eqref{eq:leastdistr0}, the set $[\Lambda_{\msf A}]$ of all least disturbing channels associated to $\msf A$ possesses an essentially uniquely defined representative of special interest:  the channel $\Lambda_{\tilde{\msf A}}$ defined as in \eqref{eq:leastdistr0} by some minimal Na\u{\i}mark dilation of the essentially unique minimally sufficient representative $\tilde{\msf A}$ of the equivalence class $[\msf A]$. Proposition \ref{prop:MinLeast} implies that $\Lambda_{\tilde{\msf A}}$ is a minimal representative of $[\Lambda_{\msf A}]$ in the sense that it can withstand no added `quantum noise', i.e.,\ if $\Lambda_{\tilde{\msf A}}$ remains invariant upon concatenation with some channel $\Gamma$, the channel $\Gamma$ must leave all the decomposable states $\sigma=\bigoplus_{j\in\Omega_{\msf A}}\sigma_j$, $\sigma_j\in\mc S\big(\msf P(j)\mc M\big)$, invariant. An equivalent formulation of Proposition \ref{prop:MinLeast} would be that, for any channel $\Gamma:\mc L(\mc M)\to\mc L(\mc M)$ such that $\Gamma\circ\Lambda_{\msf A}=\Lambda_{\msf A}$, $\mc D$, the algebra of those operators on $\mc M$ commuting with $\msf P$, is contained in the fixed-point space of $\Gamma^*$.

Suppose that $\dim{\hil}<\infty$. It now follows that, among the class $[\Lambda_{\msf A}]$, $\Lambda_{\tilde{\msf A}}$ has the Schr\"odinger output Hilbert space with the lowest dimensionality, which is hence
\begin{equation}
\sum_{k\in\Omega_{\tilde{\msf A}}}{\rm rank}\,\tilde{\msf A}(k) \, .
\end{equation}
Indeed, let $\Lambda\in[\Lambda_{\msf A}]$ meaning that there is an observable $\msf B:\Omega_{\msf B}\to\mc L(\hil)$, $\msf B\simeq\msf A$, such that $\Lambda=\Lambda_{\msf B}$, where $\Lambda_{\msf B}$ is defined by some Na\u{\i}mark dilation $(\mc N,\msf Q,K)$ of $\msf B$ as in \eqref{eq:leastdistr}. It follows that the Schr\"odinger-output dimension of $\Lambda$ is $\dim{\mc N}\geq\sum_{i\in\Omega_{\msf B}}{\rm rank}\,\msf B(i)$. Because of the construction for $\tilde{\msf A}$ given in Section \ref{subsec:minsuff}, there is a function $f:\Omega_{\msf B}\to\Omega_{\tilde{\msf A}}$ such that $\tilde{\msf A}(k)=\sum_{i\in f^{-1}(k)}\msf B(i)$. We may again freely assume that $\msf B(i)\neq0$ for all $i\in\Omega_{\msf B}$. Since, for all $i\in f^{-1}(k)$ and all $k\in\Omega_{\tilde{\msf A}}$, $\msf B(i)$ is a positive multiple of $\tilde{\msf A}(k)$, it now follows that ${\rm rank}\,\msf B(i)={\rm rank}\,\tilde{\msf A}(k)$ for all $i\in f^{-1}(k)$ and all $k\in\Omega_{\tilde{\msf A}}$, implying 
that 
\begin{equation}
\sum_{k\in\Omega_{\tilde{\msf A}}}{\rm rank}\,\tilde{\msf A}(k)\leq\sum_{i\in\Omega_{\msf B}}{\rm rank}\,\msf B(i)\leq\dim{\mc N} \, .
\end{equation}

Since $\tilde{\msf A}$ is a POVM and it satisfies $\sum_{k}\tilde{\msf A}(k)=\id$, we have $\sum_k \tr{\tilde{\msf A}(k)} = \dim\hi$. 
As $\tr{\tilde{\msf A}(k)} \leq {\rm rank}\,\tilde{\msf A}(k)$ and the equality holds if and only if $\tilde{\msf A}(k)$ is a projection, we conclude that \emph{the minimal output dimension of $\msf A$ is strictly greater than the input dimension unless the minimally sufficient representative $\tilde{\msf A}$} is sharp. 

%%%%%%%%%%%%%%%%%%%%%%%
\section{Sequential measurements of two observables}
%%%%%%%%%%%%%%%%%%%%%%%

As already stated, whenever $\msf A$ is an observable, any least disturbing channel $\Lambda_\Ao$ is $\mfr O_\Ao$-universal for $\Ao$ since the least disturbing class $[\Lambda_\Ao]$ of channels is exactly the set of $\mfr C_\Ao$-universal channels. It follows that all $\mfr C_\Ao$-universal channels (i.e.\ those within the class $[\Lambda_\Ao]$) are $\mfr O_\Ao$-universal. It might be that there exists an $\mfr O_\Ao$-universal channel which is not concatenation equivalent with $\Lambda_\Ao$, but we have no example of such an observable $\Ao$. However, in this section, we see that if $\Ao$ is extreme, these universality properties coincide.

%%%%%%%%%%%%%%%%%%%%%%%%%%
\subsection{Universal channels for extreme observables}\label{subsec:extuniversal}
%%%%%%%%%%%%%%%%%%%%%%%%%

In this section, we show that, in order to be $\mfr O_\Ao$-universal for an extreme observable $\msf A$, a channel has to be in the least disturbing class $[\Lambda_{\msf A}]$. First let us recall the physical meaning of extreme observables and characterizations of extremality.

When we have two measurement settings in our disposal, we may classically mix the observables implemented by these measurements. One can, e.g.,\ toss a (biased) coin and carry out one of the measurements conditioned by the result of the coin tossing. To formalize what this means, let us fix an outcome set $\Omega\subset\N$, the set of outcomes of the observables we concentrate on, and the system Hilbert space $\hil$. Let us denote by ${\bf Obs}(\Omega,\hil)$ the set of observables $\msf A:\Omega\to\mc L(\hil)$ with the fixed set of outcomes. The scenario above is an indication of the fact that ${\bf Obs}(\Omega,\hil)$ is a convex set where the rule of convex combinations is given by the following definition.

\begin{definition}
Suppose that $\msf A,\,\msf B\in{\bf Obs}(\Omega,\hil)$ and $t\in[0,1]$. We define the {\it mixing} $t\msf A+(1-t)\msf B\in{\bf Obs}(\Omega,\hil)$ of $\msf A$ and $\msf B$ as
$$
\big(t\msf A+(1-t)\msf B\big)(j)=t\msf A(j)+(1-t)\msf B(j),\qquad j\in\Omega.
$$
\end{definition}

One can inductively define similar mixings
$$
t_1\msf A_1+\cdots+t_n\msf A_n\in{\bf Obs}(\Omega,\hil),\quad \msf A_1,\ldots,\,\msf A_n\in{\bf Obs}(\Omega,\hil)
$$
whenever $t_k\geq0$, $k=1,\ldots,\,n$, $t_1+\cdots+t_n=1$. Mixing is a mathematical description of the classical mixing of measurement set-ups outlined above.

An observable $\Ao$ is \emph{extreme} if it cannot be expressed as a proper convex mixture of some other observables. To give a proper definition of what this means, consider the set ${\bf Obs}(\Omega,\hil)$ of observables operating on a system described by the Hilbert space $\hil$ and having $\Omega\subset\N$ as their outcome set. An element $\msf B\in{\bf Obs}(\Omega,\hil)$ is an {\it extreme point of ${\bf Obs}(\Omega,\hil)$} if it is an extreme point of the set ${\bf Obs}(\Omega,\hil)$ in the usual convex geometric sense: condition $\msf B=t\msf B_1+(1-t)\msf B_2$ with some $\msf B_1,\,\msf B_2\in{\bf Obs}(\Omega,\hil)$ and some $t\in(0,1)$ implies $\msf B_1=\msf B=\msf B_2$.

We simply say that an observable $\msf A:\Omega_{\msf A}\to\mc L(\hil)$ is {\it extreme} if $\msf A$ is an extreme point of the set ${\bf Obs}(\Omega_{\msf A},\hil)$. This definition formalizes the notion that an extreme observable $\msf A$ cannot be obtained through non-trivial mixing of other observables with the same value space as that of $\msf A$.

Let $(\mc M,\msf P,J)$ be a minimal Na\u{\i}mark dilation of an observable $\Ao:\Omega_{\Ao}\to\mc L(\hil)$. Denote by $\mc D\subset\mc L(\mc M)$ the algebra of (block-diagonal) operators $D\in\mc L(\mc M)$ such that $D\msf P(j)=\msf P(j)D$ for all $j\in\Omega_{\Ao}$. One can show \cite{Arveson69, Pellonpaa11} that $\Ao$ is extreme if and only if the map $\mc D\ni D\mapsto J^*DJ\in\mc L(\hil)$ is injective. An analogous result holds also for continuous observables.

An equivalent extremality characterization can be formulated using spectral decompositions of the effects of the observable if the Hilbert space $\hil$ is finite dimensional. Suppose that, for any $j\in\Omega_{\Ao}$, we have the spectral decomposition
\begin{equation}
\Ao(j)=\sum_{k=1}^{r_j}\kb{d_{j,k}}{d_{j,k}},
\end{equation}
where $r_j$ is the rank of $\Ao(j)$ (we only consider outcomes $j$ corresponding to non-zero $\Ao(j)$). The observable $\Ao$ is extreme if and only if the set $\{\kb{d_{j,k}}{d_{j,\ell}}\,|\,k,\,\ell=1,\ldots,\,r_j,\ j\in\Omega_{\Ao}\}$ is linearly independent \cite{DaLoPe05, Parthasarathy99}. One can derive this extremality characterization by applying the earlier algebraic characterization to the minimal Na\u{\i}mark dilation of Example \ref{ex:naimark}.

If $\Ao$ is extreme, then the set $\{ \Ao(j): j\in\Omega_\Ao \}$ is linearly independent. The linear independence follows immediately from the second characterization of extreme observables, but can also be observed in the following direct way \cite{DaLoPe05}. Assume the converse: the set $\{ \Ao(j): j\in\Omega_\Ao \}$ is linearly dependent. Then $\sum_j c_j \Ao(j)=0$ for some real numbers $c_j$ such that  $c \equiv \sum_j \mo{c_j}\neq 0$. We can thus write $\Ao$ as a mixture
\begin{equation}
\Ao = \half \Ao_+ + \half \Ao_- \, , 
\end{equation}
where $\Ao_{\pm}(j) = (1\pm c_j/c) \Ao(j)$ are two observables different than $\Ao$. Therefore, $\Ao$ is not extreme. This observation implies that, for an extreme observable $\Ao$, $\Ao$ is the essentially unique minimally sufficient representative of the post-processing equivalence class $[\Ao]$. However, not all minimally sufficient observables are extreme since not all post-processing equivalence classes of POVMs contain extreme representatives. Sharp observables are extreme but an extreme POVM does not have to be sharp; see examples, e.g.,\ in \cite{HaPe17,Pellonpaa11}.

The following proposition states that for $\mfr O_\Ao$-universality for an extreme observable nothing less than least disturbing channels is enough. In other words, for an extreme $\Ao$, $\mfr C_\Ao$-universality and $\mfr O_\Ao$-universality are equivalent properties. The proof follows the ideas of those of \cite[Theorem 3]{JePe06}.

\begin{proposition}\label{prop:extuniversal}
Let $\msf A:\Omega_{\msf A}\to\mc L(\hil)$ be an extreme observable. If $\Phi:\mc S(\hil)\to\mc S(\mc K)$ is an $\mfr O_\Ao$-universal channel for $\msf A$, then $\Phi\in[\Lambda_{\msf A}]$.
\end{proposition}

\begin{proof}
Let $\Phi:\mc S(\hil)\to\mc S(\mc K)$ be an $\mfr O_\Ao$-universal channel for $\msf A$ with the Schr\"odinger output space $\mc K$. In this proof, we express channels and other normal CP-maps mainly in the Heisenberg picture. Fix a minimal Na\u{\i}mark dilation $(\mc M,\msf P,J)$ for $\msf A$ and denote by $\mc D$ the commutant of the range of $\msf P$. There is a channel $\tilde\Gamma^*:\mc L(\mc K)\to\mc L(\mc M)$ such that $\Phi^*=\Lambda_{\msf A}^*\circ\tilde\Gamma^*$. Define the channel $\mathbb E_{\msf P}^*:\mc L(\mc M)\to\mc D$,
\begin{equation}
\mathbb E_{\msf P}^*(B)=\sum_{j\in\Omega_{\msf A}}\msf P(j)B\msf P(j),\quad B\in\mc L(\mc M) \, , 
\end{equation}
and define $\Gamma$, $\Gamma^*:=\mathbb E_{\msf P}^*\circ\tilde{\Gamma}^*$. Using the extremality of $\msf A$, one immediately sees that $\Gamma$ is the unique channel with Heisenberg output in $\mc D$ such that $\Phi=\Gamma\circ\Lambda_{\msf A}$.

We may assume that the support projection of $\Gamma$ (see Lemma \ref{lemma:suppPhi-Sapfo}) is $\id_{\mc K}$. Indeed, if this is not the case, i.e.,\ the support $R\neq\id_{\mc K}$, define $R\mc K=:\tilde{\mc K}$. Suppose that $\msf B$ is an observable jointly measurable with $\msf A$. Thus, there is an observable $\msf B'$ in $\mc K$ such that $\msf B=\Phi^*\circ\msf B'$. Define $\tilde{\msf B}:=R\msf B'(\cdot)R$ which we view as an observable in $\tilde{\mc K}$. Using item (i) of Lemma \ref{lemma:suppPhi-Sapfo}, we find
\begin{equation}
\msf B=\Lambda_{\msf A}^*\circ\Gamma^*\circ\msf B'=\Lambda_{\msf A}^*\circ\Gamma^*\circ\tilde{\msf B}=\Phi^*|_{\mc L(\tilde{\mc K})}\circ\tilde{\msf B}
\end{equation}
implying that the channel $\tilde{\Phi}$, $\tilde{\Phi}^*=\Phi^*|_{\mc L(\tilde{\mc K})}$ is also $\mfr O_\Ao$-universal. Let us thus assume that $R=\id_{\mc K}$.

Define the set
$$
\mc A:=\{A\in\mc L(\mc K)\,|\,\Gamma^*(A^*A)=\Gamma^*(A)^*\Gamma^*(A)\ \mr{and}\ \Gamma^*(AA^*)=\Gamma^*(A)\Gamma^*(A)^*\}.
$$
This set is a von Neumann algebra, the algebra of those $A\in\mc L(\mc K)$ such that $\Gamma^*(AB)=\Gamma^*(A)\Gamma^*(B)$ and $\Gamma^*(BA)=\Gamma^*(B)\Gamma^*(A)$ for all $B\in\mc L(\mc K)$ \cite[Section 9.2]{MTOA81}. Since $\Gamma^*|_{\mc A}$ is a normal *-homomorphism, the image space $\Gamma^*(\mc A)\subset\mc D$ is a von Neumann algebra as well. Let $C\in\mc A$ be such that $\Gamma^*(C)=0$. We have $\Gamma^*(C^*C)=\Gamma^*(C)^*\Gamma^*(C)=0$ implying that, since the support of $\Gamma$ is $\id_{\mc K}$, $C^*C=0$. Thus $C=0$ and $\Gamma^*|_{\mc A}$ is injective. Thus $\Gamma$ is a *-isomorphism from $\mc A$ onto $\Gamma^*(\mc A)$.

Next we show that $\Gamma^*(\mc A)=\mc D$. Pick a positive element $D\in\mc D$. For any $\varepsilon>0$, there is a PVM $\msf P_\varepsilon:\N\to\mc D$ and a function $d_\varepsilon:\N\to[0,\|D\|]$ such that, denoting $D_\varepsilon:=\sum_{n\in\N}d_\varepsilon(n)\msf P_\varepsilon(n)$,
\begin{equation}
\|D-D_\varepsilon\|<\varepsilon.
\end{equation}
Such approximate decompositions can be obtained for $D$, e.g.,\ by suitable discretizations of the spectral measure associated with $D$; the domain of $\msf P_\varepsilon$ can even be assumed to be finite. 
Define the observable $\msf B_\varepsilon:\N\to\mc L(\hil)$, $\msf B_\varepsilon(n)=J^*\msf P_\varepsilon(n)J$ for all $n\in\N$, i.e.,\ $\msf B_\varepsilon=\Lambda_{\msf A}^*\circ\msf \varepsilon$. 
Thus $\msf B_\varepsilon$ is jointly measurable with $\msf A$ and there is an observable $\msf B'_\varepsilon:\N\to\mc L(\mc K)$ such that $\msf B_\varepsilon=\Phi^*\circ\msf B'_\varepsilon$. 
Hence, for any $n\in\N$,
\begin{equation}
J^*\Gamma^*\big(\msf B'_\varepsilon(n)\big)J=\msf B_\varepsilon(n)=J^*\msf P_\varepsilon(n)J \, , 
\end{equation}
and, using the extremality of $\msf A$, $\Gamma^*\circ\msf B'_\varepsilon=\msf P_\varepsilon$. Using the item (iii) of Lemma \ref{lemma:suppPhi-Sapfo} and the fact that the support of $\Gamma$ is $\id_{\mc K}$, we find that $\msf B'_\varepsilon$ is a PVM. Define $E=\sum_{n\in\N}\sqrt{d_\varepsilon(n)}\msf B'_\varepsilon(n)$. It follows easily using the functional calculus of the PVM $\msf B'_\varepsilon$ that
\begin{equation}
\Gamma^*(E^2)=D_\varepsilon=\Gamma^*(E)^2 \, .
\end{equation}
Thus $E\in\mc A$ and $D_\varepsilon=\Gamma^*(E^2)\in\Gamma^*(\mc A)$. This holds for any $\varepsilon>0$, and since $\Gamma^*(\mc A)$ is (as a von Neumann algebra) a $C^*$-algebra, $D\in\Gamma(\mc A)$. Since $\Gamma^*(\mc A)$ contains all the positive elements of $\mc D$, $\mc D\subset\Gamma^*(\mc A)$. The converse is trivial, so that $\Gamma^*(\mc A)=\mc D$.

Denote by $\tilde\Psi$ the channel such that $\tilde\Psi^*:\mc D\to\mc A$ the inverse of the *-isomorphism $\Gamma^*|_{\mc A}:\mc A\to\mc D$. Further, define $\Psi$, $\Psi^*=\tilde\Psi^*\circ\mathbb E_{\msf P}^*:\mc L(\mc M)\to\mc A\subset\mc L(\mc K)$. This map is unital and hence a channel. Since, for all $B\in\mc L(\mc M)$, $\mathbb E_{\msf P}^*(B)\in\mc D$ and using the definition of $\tilde{\Psi}$, we have
\begin{eqnarray*}
(\Phi^*\circ\Psi^*)(B)&=&(\Lambda_{\msf A}^*\circ\Gamma^*\circ\tilde{\Psi}^*\circ\mathbb E_{\msf P}^*)(B)=(\Lambda_{\msf A}^*\circ\Gamma^*\circ\tilde{\Psi}^*)\big(\mathbb E_{\msf P}^*(B)\big)\\
&=&\Lambda_{\msf A}^*\big(\mathbb E_{\msf P}^*(B)\big)=\Lambda_{\msf A}^*(B)
\end{eqnarray*}
for all $B\in\mc L(\mc M)$. Hence, $\Lambda_{\msf A}\psleq\Phi$, implying $\Phi\in[\Lambda_{\msf A}]$.
\end{proof}

Note that the use of the $\varepsilon$-treatment in the above proof is there only because we concentrate in this paper on discrete observables. If we allowed for continuous observables, we could use the true spectral measure of the positive $D\in\mc D$ in the above proof instead of its $\varepsilon$-approximate discretization $\msf P_\varepsilon$.

%%%%%%%%%%%%%%%%%%%%
\section{Summary and discussion}
%%%%%%%%%%%%%%%%%%%%

We have discussed ways to carry out quantum devices -- observables and channels -- compatible with a fixed observable $\msf A$ in a sequential setting where $\msf A$ is measured first. The observable $\msf A$ can be measured in such a way that (i) we may realize any device compatible with $\msf A$ after the said measurement or (ii) we may realize any observable compatible with $\msf A$ after the said measurement. The crucial part of the measurement of $\msf A$ for these properties is the channel (unconditioned state transformation) induced by the measurement.

When the measurement of $\msf A$ has the property (i), we say that the corresponding channel is $\mfr C_\Ao$-universal. We have characterized the class of all those channels that possess this universality property: this class is the class of least disturbing channels for any observable $\tilde{\msf A}$ in the post-processing equivalence class of $\msf A$. All channels in this class possess the universality property associated with the property (ii) above, which we call $\mfr O_\Ao$-universality. 
However, it is not clear whether for property (ii) something strictly less is already enough than for property (i). We have the partial result stating that, if $\msf A$ is extreme, $\mfr C_\Ao$-universality and $\mfr O_\Ao$-universality are equivalent properties of channels. Giving a definitive answer to the question regarding the relationship between these two universality properties for any observable $\msf A$ shall be a subject of future studies.

One can generalize the universality properties described above in the following way: Consider a class $\mfr D$ of devices compatible with a fixed observable $\msf A$, i.e.,\ $\mfr D\subset\mfr C_\Ao$. In our framework we may consider observables as channels so that $\mfr O_\Ao\subset\mfr C_\Ao$. We say that a channel $\Psi$ is {\it $\mfr D$-universal for $\msf A$} if $\Psi$ is compatible with $\msf A$ and, for any $\mc D\in\mfr D$, there is some device $\mc D'$ such that $\mc D=\mc D'\circ\Psi$. This means that there is a way to measure $\msf A$ such that any device from $\mfr D$ can be realized after the said measurement of $\msf A$.

Typically in an experimental setting, it is not possible to carry out arbitrary channels or observables after a measurement of the first observable $\msf A$. There is typically a restriction on the value spaces of the subsequent measurements or restriction on the output dimension of the subsequent channels. Another typical situation is that the measurement setting in our disposal is able to realize only observables and channels that reflect some symmetries, e.g.,\ in the form of covariance with respect to unitary representations. Thus, concentrating on $\mfr D$-universal channels for the observable $\msf A$ reflects our inability to realize the whole of $\mfr C_\Ao$ or $\mfr O_\Ao$ forcing us to find the least disturbing channels which can be reached with the measurement settings in our disposal. Universality properties corresponding to restrictions on output spaces of the subsequent measurement processes and to covariance requirements thus have a clear physical meaning and provide mathematically interesting problems for future study.

\section*{Acknowledgements}

E.H. acknowledges financial support from the Japan Society for the Promotion of Science (JSPS) as an overseas postdoctoral fellow at JSPS. T.H. acknowledges financial support from the Horizon 2020 EU collaborative projects QuProCS (Grant Agreement No. 641277) and the Academy of Finland (Project no. 287750).

\newpage

%%%%%%%%%%%%%%%%%%%%
\section*{Appendix}
%%%%%%%%%%%%%%%%%%%%

In the proofs of Corollary \ref{lemma:fiberwise_id} and Proposition \ref{prop:extuniversal} we need a couple of technical lemmata which we state and prove in this appendix.

%%%%%%%%%%%%%%%%%%%%
\subsection*{Equivalence with the identity channel}
%%%%%%%%%%%%%%%%%%%%

In this subsection, we characterize quantum channels which are concatenation equivalent with the identity channel, a result needed in the proof of Theorem \ref{theor:leastdistreq}. First we need to discuss the notion of conjugate channels.

\begin{definition}
Let $\Lambda:\mc S(\hil)\to\mc S(\mc K)$ be a channel and $(\mc N,V)$ some Stinespring dilation for $\Lambda$. The channel $\Lambda^c:\mc S(\hil)\to\mc S(\mc N)$, $\Lambda^c(\varrho)=\mr{tr}_{\mc K}[V\varrho V^*]$, $\varrho\in\mc S(\hil)$, is called a {\it conjugate channel of $\Lambda$}.
\end{definition}

Although every dilation of a channel $\Lambda:\mc S(\hil)\to\mc S(\mc K)$ defines its own conjugate channel, they are all mutually post-processing equivalent. To see this, fix a minimal Stinespring dilation $(\mc N,V)$ and some other dilation $(\tilde{\mc N},\tilde{V})$ for $\Lambda$, and denote the conjugate channel defined by $(\mc N,V)$ by $\Lambda^c$ and that defined by $(\tilde{\mc N},\tilde{V})$ by $\tilde{\Lambda}^c$. Let $W:\mc N\to\tilde{\mc N}$ be an isometry such that $(\id_{\mc K}\otimes W)V=\tilde{V}$. Defining the channels $\Phi:\mc S(\mc N)\to\mc S(\tilde{\mc N})$ and $\iota:\mc S(\tilde{\mc N})\to\mc S(\mc N)$,
\begin{equation}
\begin{array}{rcll}
\Phi(\varrho)&=&W\varrho W^*,&\varrho\in\mc S(\mc N)\\
\iota(\tilde{\varrho})&=&W^*\tilde{\varrho}W+\tr{(\id_{\tilde{\mc N}}-WW^*)\tilde\varrho}\sigma,&\tilde\varrho\in\mc S(\tilde{\mc N})
\end{array}
\end{equation}
with some $\sigma\in\mc S(\mc N)$, one finds that $\Phi\circ\Lambda^c=\tilde{\Lambda}^c$ and $\iota\circ\tilde{\Lambda}^c=\Lambda^c$.

Let us fix a separable Hilbert space $\hil$. Let us denote by ${\bf SWAP}$ the set of channels $\Gamma:\mc S(\hil)\to\mc S(\mc N)$ (with varying separable Hilbert spaces $\mc N$) such that $\Gamma(\varrho)=\tr{\varrho}\sigma$ with some positive trace-1 operator $\sigma$ on $\mc N$. Denote also the set of observables $\msf B:\Omega_{\msf B}\to\mc L(\hil)$, $\Omega_{\msf B}\subset\N$, such that $\msf B(k)=p(k)\id_\hil$ with some probability distribution $p:\Omega_{\msf B}\to[0,1]$, by ${\bf TRIV}$. It is simple to see that ${\bf SWAP}$ and ${\bf TRIV}$ are both single post-processing equivalence classes.

\begin{lemma}\label{lemma:eqid}
For a channel $\Lambda:\mc S(\hil)\to\mc L(\mc K)$, the following are equivalent:
\begin{itemize}
\item[(i)] $\Lambda \simeq \mr{id}_{\mc S(\hil)}$
\item[(ii)] If $\Lambda \psleq \Lambda_{\Ao}$ for some observable $\Ao$ with $\Lambda_\Ao$ defined as in \eqref{eq:leastdistr0} with some N\u{\i}mark dilation of $\msf A$, then $\Ao\in{\bf TRIV}$. 
\item[(iii)] $\Lambda^c\in{\bf SWAP}$ for some conjugate channel $\Lambda^c$ of $\Lambda$
\item[(iv)] $\Lambda(\varrho) =\sum_n p_n V_n\varrho V_n^*$ for a probability distribution $\{p_n\}_n$ and some isometries $V_n:\hil\to\mc K$ satisfying 
$V_n^* V_m =\delta_{nm}\id$. (That is, $V_n \hil \perp V_m \hil$.)
\end{itemize}
\end{lemma}

\begin{proof}
We first prove (i)$\Rightarrow$(ii). Assume that $\Lambda$ is compatible with an observable $\Ao$. If $\mr{id}_{\mc S(\hil)}\simeq\Lambda$, then also $\mr{id}_{\mc S(\hil)}$ is compatible with $\Ao$, implying that $\msf A\in{\bf TRIV}$.

Assume now (ii). Fix a minimal Stinespring dilation $(\mc N,V)$ for $\Lambda$. For any $C\in\mc L(\mc N)$, $0\leq C\leq\id_{\mc N}$, the binary observable $\msf A:\{0,1\}\to\mc L(\hil)$, $\msf A(0)=V^*(\id_{\mc N}\otimes C)V$, is compatible with $\Lambda$ implying that there is $p_C\in[0,1]$ such that $p_C\id_\hil=\msf A(0)=V^*(\id_{\mc K}\otimes C)V$. This means that the conjugate channel $\Lambda^c$ defined by the dilation $(\mc N,V)$ is in ${\bf SWAP}$.

Let us next show (iii)$\Rightarrow$(iv). Let $(\mc N,V)$ be a Stinespring dilation for $\Lambda$, $\Lambda^c$ the associated conjugate channel, and $\sigma$ a positive trace-1 operator on $\mc N$ such that $\Lambda^c(\varrho)=\tr{\varrho}\sigma$. Let $\mc N_0:=\mr{ran}\,\sigma^{1/2}$ and numbers $p_n\in(0,1]$ and unit vectors $b_n\in\mc N_0$, $n=1,\ldots,\,\mr{rank}\,\sigma$, constitute a spectral decomposition
\begin{equation}
\sigma=\sum_{n=1}^{\mr{rank}\,\sigma}p_n|b_n\>\<b_n| \, .
\end{equation}
If $\sigma$ is not of full rank, complete the set $\{b_n\}_{n=1}^{\mr{rank}\,\sigma}$ into an orthonormal basis $\{b_n\}_{n=1}^{\dim{\mc N}}$ of $\mc N$ and set $p_n=0$ whenever $n>\mr{rank}\,\sigma$. Define the vector $\psi=\sum_n\sqrt{p_n}b_n\otimes b_n$. We may now define the minimal Stinespring dilation $(\hil\otimes\mc N_0,V_0)$ for $\Lambda^c$, where $V_0\fii=\fii\otimes\psi$, $\fii\in\hil$. Note that, for notational reasons, we define the dilation in the form $\Lambda^c(\varrho)=\mr{tr}_{\mc K}[V_0\varrho V_0^*]$, $\varrho\in\mc S(\hil)$. There is now an isometry $W:\hil\otimes\mc N_0\to\mc K$ such that $(W\otimes\id_{\mc N})V_0=V$. Define the operators $R_n:\mc K\otimes\mc N\to\mc K$, $R_n(\eta\otimes\xi)=\<b_n|\xi\>\eta$, $\eta\in\mc K$, $\xi\in\mc N$. Note that $R_m^*R_n=\id_{\mc K}\otimes|b_m\>\<b_n|$. For $n\leq\mr{rank}\,\sigma$, define $V_n=p_n^{-1/2}R_nV$. It follows that
\begin{eqnarray*}
\<V_m\fii|V_n\fii\>&=&\frac{1}{\sqrt{p_mp_n}}\<V\fii|(\id_{\mc K}\otimes|b_m\>\<b_n|)V\fii\>\\
&=&\frac{1}{\sqrt{p_mp_n}}\<b_n|\Lambda^c(|\fii\>\<\fii|)b_m\>=\delta_{m,n}\|\fii\|^2,
\end{eqnarray*}
implying that $V_m^*V_n=\delta_{m,n}\id_\hil$. Moreover, $V\fii=\sum_n\sqrt{p_n}V_n\fii\otimes b_n$ for every $\fii\in\hil$ proving (iv).

Assume now (iv). Define the channel $\Gamma:\mc S(\mc K)\to\mc L(\hil)$,
\begin{equation}
\Gamma(\tau)=\sum_nV_n^*\tau V_n+\tr{\Big(\id_{\mc K}-\sum_nV_nV_n^*\Big)\tau}\varrho_0,\quad\tau\in\mc S(\mc K) \, , 
\end{equation}
where $\varrho_0$ is a positive trace-1 operator on $\hil$. Recall that the projections $V_nV_n^*$ are mutually orthogonal. Straight calculation shows that $\Gamma\circ\Lambda=\mr{id}_{\mc S(\hil)}$.
\end{proof}

%%%%%%%%%%%%%%%%%%%%
\subsection*{Support projection}\label{subsec:supp}
%%%%%%%%%%%%%%%%%%%%

In the proof of Proposition \ref{prop:extuniversal}, the notion of a support projection of a quantum channel is needed. This support projection is defined in the lemma below and some of its important properties are stated.

\begin{lemma}\label{lemma:suppPhi-Sapfo}
Let $\hil$ and $\mc K$ be Hilbert spaces and $\Lambda:\mc S(\hil)\to\mc S(\mc K)$ be a channel. There is a unique projection $R$ on $\mc K$ such that $\Lambda^*(R)=\id_\hil$ and, whenever $Q\leq R$ is a projection, $\Lambda^*(Q)=\id_\hil$ implies $Q=R$. Moreover, when the projection $R$ is as above,
\begin{itemize}
\item[(i)] $\Lambda^*(B)=\Lambda^*(RB)=\Lambda^*(BR)=\Lambda^*(RBR)$ for all $B\in\mc L(\mc K)$,
\item[(ii)] for a positive $E\in\mc L(\mc K)$, $\Lambda^*(E)=0$ implies $RER=0$, and
\item[(iii)] whenever $E\in\mc L(\mc K)$ is positive and $\Lambda^*(E)$ is a projection, then $RER$ is a projection as well and $RE=ER$.
\end{itemize}
\end{lemma}

\begin{proof}
The first part of the Lemma is well-known, and proofs for the claims in items (i) and (ii) can be found, e.g.,\ in \cite[Section 10.8]{QM16}. Let now $E$ be an effect such that $\Lambda^*(E)$ is a projection. Using the Schwarz inequality, one finds
\begin{equation}
\Lambda^*(E)=\Lambda^*(RER)=\Lambda^*(RER)^2\leq\Lambda^*(RERER)=\Lambda^*(ERE)
\end{equation}
implying $\Lambda^*(E-ERE)\leq0$. Since $E\geq ERE$ and $\Lambda^*$ is positive, we have $\Lambda^*(E-ERE)=0$, and thus, using the above result, one obtains $RER=RERER$, i.e.,\ $RER$ is a projection. Note that if $\Lambda^*(E)\neq0$, also $RER$ is non-zero since $\Lambda^*(E)=\Lambda^*(RER)$. It is simple to show that $PAP$ is a projection for any projection $P$ and a positive operator $A$ only if $P$ and $A$ commute. Therefore, $E$ and $R$ commute.
\end{proof}

%%%%%%%%%%%%%%%%%%%%%%%%%%%%%%%%%
%%%%%%%%%%%%%%%%%%%%%%%%%%%%%%%%%

\newpage

\providecommand{\bysame}{\leavevmode\hbox to3em{\hrulefill}\thinspace}
\providecommand{\MR}{\relax\ifhmode\unskip\space\fi MR }
\providecommand{\MRhref}[2]{%
  \href{http://www.ams.org/mathscinet-getitem?mr=#1}{#2}
}
\providecommand{\href}[2]{#2}

\end{document}